\documentclass[10pt,twocolumn,letterpaper]{article}
\usepackage[pagenumbers]{cvpr} % 

\usepackage[dvipsnames]{xcolor}

\usepackage{amsthm}

 {\everymath{\displaystyle\everymath{}}\array}%
 {\endarray}
\everymath{\displaystyle\everymath{}}

\newtheorem{theorem}{Theorem}[section]

\newtheorem{claim}[theorem]{Claim}

\def\A{\mathbf{A}}
\def\B{\mathbf{B}}

\def\D{\mathbf{D}}

\def\I{\mathbf{I}}

\def\P{\mathbf{P}}

\def\S{\mathbf{S}}

\def\U{\mathbf{U}}
\def\V{\mathbf{V}}
\def\W{\mathbf{W}}

\def\e{\mathbf{e}}

\def\g{\mathbf{g}}

\def\x{\mathbf{x}}
\def\y{\mathbf{y}}
\def\z{\mathbf{z}}

\def\bGamma{\boldsymbol{\Gamma}}

\def\bLambda{\boldsymbol{\Lambda}}

\def\bepsilon{\boldsymbol{\epsilon}}

\def\Exp{\mathbb{E}}

\def\0{\mathbf{0}}
\def\1{\mathbf{1}}

\newcommand{\tomt}[1]{\textcolor{black}{#1}}
\newcommand{\tomtb}[1]{\textcolor{black}{#1}}
\newcommand{\tomtr}[1]{\textcolor{black}{#1}}
\newcommand{\tomtm}[1]{\textcolor{black}{#1}}
\newcommand{\tomtf}[1]{\textcolor{black}{#1}}

\usepackage{color}
\usepackage{algorithm}
\usepackage{graphicx}
\usepackage{subcaption}
\usepackage{algpseudocode}
\usepackage{wrapfig}
\usepackage{diagbox}
\definecolor{cvprblue}{rgb}{0.21,0.49,0.74}
\usepackage[pagebackref,breaklinks,colorlinks,citecolor=cvprblue]{hyperref}

\def\paperID{13018} % 
\def\confName{CVPR}
\def\confYear{2024}

\title{Image Restoration by Denoising Diffusion Models with Iteratively Preconditioned Guidance}

\author{Tomer Garber\\
Open University of Israel\\
{\tt\small tomergarber@gmail.com}
\and
Tom Tirer\\
Bar-Ilan University, Israel\\
{\tt\small tirer.tom@gmail.com}
}

\begin{document}
\maketitle

\begin{abstract}
Training deep neural networks has become a common approach for addressing image restoration problems.
An alternative % 
for training a ``task-specific" network for each observation model is to use pretrained deep denoisers for imposing only the signal's prior within iterative algorithms, without additional training. % 
Recently, \tomtb{a sampling-based} variant of this approach has become % 
popular with the rise of diffusion/score-based generative models. % 
Using denoisers for general purpose restoration requires guiding the iterations to ensure agreement of the signal with the observations.
In low-noise settings, guidance that is based on back-projection (BP) has been shown to be a promising strategy % 
(used recently % 
also under the names 
``pseudoinverse" or ``range/null-space" guidance).
However, the presence of noise in the observations hinders the % 
gains from
this approach. 
In this paper, we propose a novel guidance technique, based on % 
preconditioning
that allows traversing from BP-based guidance to least squares based guidance along the restoration scheme.
The proposed approach is robust to noise while still having much simpler  implementation than alternative methods (e.g., it does not require SVD or a large number of iterations).
\tomtb{We use it within both an optimization scheme and a sampling-based scheme}, and 
demonstrate its advantages over existing methods 
for image deblurring and super-resolution. 
\end{abstract}

\begin{figure}
    \centering
    \begin{subfigure}[h]{0.32\columnwidth}
        \centering
        \includegraphics[width=2.7cm, height=2.7cm]{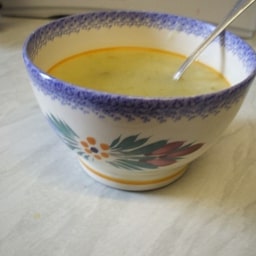}
    \end{subfigure} % 
    \begin{subfigure}[h]{0.32\columnwidth}
        \centering
        \includegraphics[width=2.7cm, height=2.7cm]{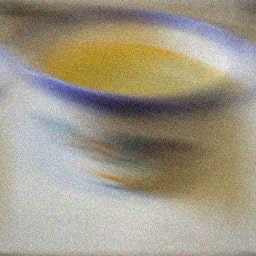}
    \end{subfigure}\\ % 
    \begin{subfigure}[h]{0.32\columnwidth}
        \centering
        \includegraphics[width=2.7cm, height=2.7cm]{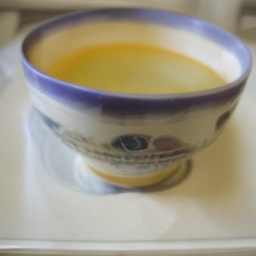}
    \end{subfigure}     
    \begin{subfigure}[h]{0.32\columnwidth}
        \centering
        \includegraphics[width=2.7cm, height=2.7cm]{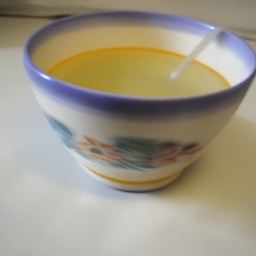}
    \end{subfigure}   
    \begin{subfigure}[h]{0.32\columnwidth} 
        \centering
        \includegraphics[width=2.7cm, height=2.7cm]{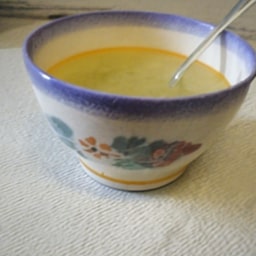}
    \end{subfigure}

    \caption{
    Motion deblurring with noise level 0.05.
    From top to bottom and left to right: original, observation, DPS \cite{chung2022diffusion}, DiffPIR \cite{zhu2023denoising} and our DDPG.
    }
    \label{fig:motion_deblur_0.05}
\end{figure}

\section{Introduction}
\label{sec:intro}

Image restoration problems appear in a wide range of applications, where the goal is to recover a high-quality image $\x^* \in \mathbb{R}^n$ from its degraded version $\y \in \mathbb{R}^m$, which can be noisy, blurry, low-resolution, etc. 
In many problems, the relation between $\y$ and $\x^*$ can be expressed using a linear observation model
\begin{align}
\label{eq:obsevation_model}
    \y = \A\x^* + \e
\end{align}
where $\A \in \mathbb{R}^{m\times n}$ is a measurement operator with $m\leq n$, and $\e\sim\mathcal{N}(\0, \sigma_e^2\I_m)$ is an additive white Gaussian noise.
For example, in the denoising task $\A$ is the identity matrix $\I_n$; in the deblurring task $\A$ is a blur operator; and in the super-resolution task $\A$ is a composition of sub-sampling and (anti-aliasing) filtering. 
Image restoration problems are % 
ill-posed: just looking for $\x\in\mathbb{R}^n$ that fits $\y$ based on the observation model does not suffice for a successful recovery. Thus, it is required to utilize prior knowledge on the nature of $\x^*$.

Nowadays, it has become common to address these problems by exhaustively training a different deep neural network (DNN) for each {\em predefined} observation model in a supervised manner. Namely,  
using \eqref{eq:obsevation_model} to generate 
training pairs $\{\y_i,\x^*_i\}$ 
and training a DNN to invert the map 
\citep{dong2015image,Sun2015LearningAC,
lim2017enhanced,zhang2017beyond}. 
However, these ``task-specific" DNNs % 
suffer from a huge performance drop when the observations at test-time mismatch (even slightly) the assumptions made in training \citep{shocher2018zero,tirer2019super,hussein2020correction}, which limits their 
applicability in many practical cases. % 

An alternative approach is to use pretrained DNNs that impose only the signal prior, while the agreement with the observations is being handled at test-time in a ``zero-shot" manner.
A successful choice of such pretrained DNNs are Gaussian denoisers, as was initially demonstrated in different ``plug-and-play" (PnP) and ``regularization by denoising" (RED) iterative schemes \citep{venkatakrishnan2013plug,romano2017little,zhang2017learning,tirer2018image}.
\tomtb{The rise of score-based generative modeling \citep{song2019generative,ho2020denoising} --- typically referred to as denoising diffusion models (DDMs) ---
whose inference/sampling ``reversed" flows
are} 
based on iteratively injecting Gaussian noise and denoising, has further increased the popularity of using iterative denoising for general purpose restoration. % 
In particular, 
operations that promote data-fidelity, similar to those used in PnP/RED, have been used  % 
within the diffusion models' iterative % 
sampling
schemes 
to guide the iterations towards an image that not only looks natural but also agrees with the observations \citep{song2021solving,
kawar2022denoising,chung2022improving,song2022pseudoinverse,wang2022zero,mardani2023variational,zhu2023denoising,chung2022diffusion,abu2022adir}. % 

Combining iterative denoising with
guidance that is based on back-projection (BP) of intermediate estimates on the subspace $\{\x: \A\x=\y \}$ (more details in Section~\ref{sec:background}) was originally proposed in \citep{tirer2018image} 
(in the context of PnP),
and has been shown to be a promising strategy \tomtb{in both the PnP and the DDM literature}  \citep{tirer2019super,zukerman2020bp,sabulal2020joint, % 
genzel2022near,savanier2023deep,chung2022improving,song2022pseudoinverse,wang2022zero,liu2023accelerating}. 
In 
the context of diffusion models 
it has been rediscovered 
under the names
``pseudoinverse" \citep{song2022pseudoinverse} or ``range/null-space" \citep{wang2022zero} guidance. 
However, 
\tomtb{in cases where $\A$ has low singular values (e.g., in image deblurring and super-resolution)}, 
the presence of noise in the observations $\y$ hinders the application of this approach: careful regularization is required % 
and the performance gains % 
tend to decrease.

\textbf{Contribution.} 
\tomt{In this paper, we identify the BP step as a % 
preconditioned
version of a least squares (LS) gradient step. % 
This inspires us to propose an iteration-dependent preconditioned guidance approach that}
essentially traverses
from the BP % 
\tomt{step}
to the LS % 
\tomt{gradient step}
along the restoration scheme.
Thus, it enjoys benefits of BP step \citep{tirer2019back,tirer2021convergence}, such as stronger consistency with $\y$ 
and accelerated convergence in early iterations, and the better robustness to noise of LS as the scheme approaches its end.
The specific % 
\tomt{trajectory} 
that we % 
devise
in this paper is also more simple and efficient to implement than other potential alternatives for BP, which 
\tomtb{either} require 
having full knowledge on the SVD of $\A$ \citep{kawar2022denoising},
\tomtb{or many neural function evaluations (NFEs) \citep{chung2022diffusion}.}

We present formal mathematical motivation for our guidance technique 
and use it within both a PnP-based scheme and its DDM sampling-based counterpart that we design.
We examine the proposed approach for image deblurring and super-resolution on the CelebA-HQ and ImageNet datasets, where we exploit the 
powerful denoisers of DDMs \citep{lugmayr2022repaint,guidedDiff}, 
and demonstrate its advantages over existing methods. 
In particular, our sampling-based reconstruction approach is flexible to the observation model, computationally convenient, and displays a good combination of perceptual quality and accuracy.

\section{Background and Related Work}
\label{sec:background}

\subsection{PnP Denoisers and Back-Projections}

Traditionally, estimating $\x^*$ from its measurements \eqref{eq:obsevation_model}
has been tackled by minimizing 
an explicit objective function
\begin{align}
\label{eq:traditional_cost}
  L(\x) = \ell(\x;\y) + s(\x),
\end{align}
which is composed of a data-fidelity term $\ell(\x;\y)$ and a signal prior term $s(\x)$. % 
The PnP concept \citep{venkatakrishnan2013plug} suggests applying a proximal algorithm to \eqref{eq:traditional_cost}, where the proximal mapping associated with the prior function $s(\cdot)$, i.e., $\mathrm{prox}_s(\x):=\mathrm{argmin}_{\z}\frac{1}{2}\|\z-\x\|_2^2+s(\z)$ \citep{moreau1965proximite}, is replaced with an off-the-shelf Gaussian denoiser $\mathcal{D}(\cdot;\sigma_t)$ ($\sigma_t$ denotes the denoiser's noise level) that need not be associated with an explicit prior function (e.g., a pretrained DNN). 
A vast amount of literature has been dedicated to designing and analyzing PnP algorithms \citep{venkatakrishnan2013plug,chan2016plug,romano2017little,zhang2017learning,tirer2018image,reehorst2018regularization,sun2019online,kamilov2023plug}.

A popular PnP scheme based on the proximal-gradient method (PGM) is given by:
\begin{align}
\label{eq:ista_step2}
{\x}_{0|t} &=  \mathcal{D}( \x_t; \sigma_t ), \\
\label{eq:ista_step1}
\x_{t-1} &=  {\x}_{0|t} - \mu_t \nabla_{\x}  \ell( {\x}_{0|t} ;\y), 
\end{align}
where $t=T,T-1,...,1$ is the iteration index (decreases over time, to match the notation in diffusion/score-based models), $\mu_t$ is a step-size, and $\sigma_T>...>\sigma_1$ is a predefined decreasing set of noise levels.
The scheme is initialized with some $\x_T$ and the final estimate is given by $\x_{0|1}$.

A typical choice of data-fidelity term is the LS objective
\begin{align}
\label{eq:ls_term}
    \ell_{LS}(\x;\y) = \frac{1}{2}\|\A\x-\y\|_2^2,
\end{align}
for which the gradient step in \eqref{eq:ista_step1} is given by
\begin{align}
\label{eq:ls_grad_step}
    \x_{t-1} &=  {\x}_{0|t} - \mu_t \A^T(\A\x_{0|t} - \y).
\end{align}
\tomtb{Note though that PnP methods 
based on LS gradient steps tend to require many iterations \citep{romano2017little,tirer2021convergence}.} 
The IDBP method \citep{tirer2018image} suggests that % 
the data-fidelity guidance, which follows the denoising operation, will be the 
back-projection (BP) of $\x_{0|t}$ onto the affine subspace $\{\x: \A\x=\y \}$:
\begin{align}
\label{eq:bp_prob}
    \x_{t-1} &= \mathrm{argmin}_{\x} \|\x-\x_{0|t}\|_2 \,\,\, \mathrm{s.t.} \,\,\, \A\x=\y \\
    \label{eq:bp}
    &= \A^\dagger\y + (\I_n-\A^\dagger\A)\x_{0|t} \\
    \label{eq:bp_grad_step}
    &= \x_{0|t} - \A^\dagger(\A\x_{0|t}-\y),
\end{align}
where $\A^\dagger$ denotes the pseudoinverse of $\A$. 

When $\A$ % 
has near-zero singular values,
it has been shown to be beneficial to replace $\A^\dagger=\A^T(\A\A^T)^\dagger$ with a regularized version $\A^T(\A\A^T+\eta\I_m)^{-1}$, where $\eta>0$ is the regularization hyperparameter
\citep{tirer2018image,tirer2019super}.
Importantly, note that % 
there are popular tasks where 
the pseudoinverse can be implemented efficiently \tomtb{(e.g., see % 
the appendix 
for FFT based implementations of this operation for deblurring and super-resolution)}. And, in general, full rank $\A\A^T$ (and $\A\A^T+\eta\I_m$ otherwise) can be typically ``inverted" using few conjugate gradient iterations \citep{hestenes1952methods}, which only require applying $\A$ and $\A^T$ and bypass the need of matrix inversion or SVD.

As pointed out in \citep{tirer2019back,tirer2021convergence}, \eqref{eq:bp_grad_step} is equivalent to a gradient step \eqref{eq:ista_step1} with step-size $\mu_t=1$ and a special choice of $\ell(\x;\y)$, which they dubbed as the ``BP term":
\begin{align}
\label{eq:bp_term}
    \ell_{BP}(\x;\y) = \frac{1}{2}\|(\A\A^T)^{-1/2}(\A\x-\y)\|_2^2.
\end{align}
These works % 
have analyzed, both empirically and theoretically, the effects of using BP steps \eqref{eq:bp_grad_step} rather than LS steps \eqref{eq:ls_grad_step} in inverse problems.
They show provable faster convergence benefits (less iterations) and, in the low-noise regime, also MSE benefits.
However, % 
it is also shown that in the presence of observation noise and % 
$\A$ with % 
small singular values,
we have that $\A^\dagger\y$ amplifies the noise. Thus a strong regularization $\eta$ is needed which hinders the advantages.

\subsection{Restoration via Guidance of DDMs}

Pretrained diffusion/score-based generative models \citep{song2019generative,ho2020denoising,song2020score} have been shown to be a powerful signal prior for image restoration.
These generative models are based on training a network that approximates the score function $\nabla_{\x_t}\log p_t(\x_t)$, where $p_t$ denotes the distribution of $\x_t$, a Gaussian noisy version of the data $\x_0 \sim p_0$ with noise level indexed by $t \in [0,T]$.
For $\x_t = \x_0 + \bepsilon_t$ with $\bepsilon_t \sim \mathcal{N}(\0,\sigma_t^2\I_n)$,
by Tweedie's formula $-\sigma_t^2 \nabla_{\x_t}\log p_t(\x_t) = \x_t - \mathbb{E}[\x_0|\x_t] =\mathbb{E}[\bepsilon_t|\x_t]$ \citep{efron2011tweedie}. Therefore, these approaches essentially attempt to learn the MMSE Gaussian denoiser or, equivalently, noise estimation.
After training, the sampling schemes are based on random Gaussian initialization and 
alternatively denoising and injection of Gaussian noise with a decreasing noise level.

For concreteness, let us focus on the popular DDPM formulation \citep{ho2020denoising}.
In DDPM, $\x_t$ is generated from a data point $\x_0$ by the ``forward flow":
\begin{align}
    \x_t = \sqrt{1-\beta_t}\x_{t-1}+\sqrt{\beta_t}\bepsilon,
\end{align}
where $\bepsilon \sim \mathcal{N}(\0,\I_n)$ and $\{ \beta_t \}_{t=1}^T$ is a predetermined ``noise schedule" obeying $0<\beta_1 < \ldots < \beta_T \leq 1$.
By properties of Gaussians, $\x_t$ can be directly generated from $\x_0$ by
\begin{align}
\label{eq:xt_x0}
    \x_t = \sqrt{\bar{\alpha}_t}\x_0+\sqrt{1-\bar{\alpha}_t}\bepsilon,
\end{align}
which includes noise variance $1-\bar{\alpha}_t$ and signal scale factor $\sqrt{\bar{\alpha}_t}$, where $\bar{\alpha}_t=\Pi_{s=1}^t \alpha_s$ and $\alpha_t = 1-\beta_t$. 
In training time, a U-net DNN $\bepsilon_\theta(\cdot;t)$ is trained to predict the noise $\bepsilon$ given $\x_t$ by minimizing an MSE loss. 
In DDPM's inference time, 
starting from a random $\x_T \sim \mathcal{N}(\0,\I_n)$, 
one step of the ``reverse flow" that is used for generating a sample from $p_0$ is given by
\begin{align}
\label{eq:ddpm_step}
    \x_{t-1} = \frac{1}{\sqrt{\alpha_t}} \left ( \x_t - \frac{1-\alpha_t}{\sqrt{1-\bar{\alpha}_t}} \bepsilon_\theta(\x_t;t) \right ) + \sqrt{\beta_t}\bepsilon_t,
\end{align}
where $\bepsilon_t$ is drawn from $\mathcal{N}(\0,\I_n)$.

A faster generation scheme, denoted DDIM \citep{song2020denoising}, is based on diverging from the Markovian guideline of DDPM. In DDIM, the step \eqref{eq:ddpm_step} is replaced with
\begin{align}
\label{eq:ddim_step}
    \x_{t-1} = \sqrt{\bar{\alpha}_{t-1}} {\x}_{0|t} + \hat{\sigma}_t \bepsilon_\theta (\x_t;t) + \tilde{\sigma}_t \bepsilon_t,
\end{align}
where $\bepsilon_t \sim \mathcal{N}(\0,\I_n)$ % 
and
\begin{align}
    \label{eq:sigma_ddim}
    \hat{\sigma}_t &= \sqrt{1-\bar{\alpha}_{t-1}-\tilde{\sigma}_t^2}, \\
    \label{eq:x0_est}
    {\x}_{0|t} &= \frac{1}{\sqrt{\bar{\alpha}_{t}}} \left (  \x_t - \sqrt{1-\bar{\alpha}_t} \bepsilon_\theta(\x_t;t) \right ).
\end{align}
The hyperparameter $\tilde{\sigma}_t$ allows trading between the levels of two noise terms: the estimated noise $\bepsilon_\theta (\x_t;t)$ and a pure stochastic noise $\bepsilon_t$.
Moreover, $\x_{0|t}$ in \eqref{eq:x0_est} is an estimate of $\x_0$ given $\x_t$, which is essentially performing Gaussian denoising: $\x_{0|t} = \frac{1}{\sqrt{\bar{\alpha}_t}}\mathcal{D}( \x_t; \sigma_t=\sqrt{1-\bar{\alpha}_t} )$, as implied by  \eqref{eq:xt_x0}.

Many
of the methods that use DDMs as priors for image restoration \citep{kawar2022denoising,song2022pseudoinverse,wang2022zero,zhu2023denoising} 
are based on the DDIM scheme \citep{song2020denoising}. % 
Yet, in order that the sampling scheme will produce an image that is not only perceptually pleasing but also agrees with the measurements $\y$ and the observation model \eqref{eq:obsevation_model} it is required to equip the iterations with some data-fidelity guidance.
This guidance is typically based on the gradient of a data-fidelity term $\ell(\x;\y)$. In other words, the iterations \eqref{eq:ddim_step} are modified into 
\begin{equation}
\label{eq:ddim_step_restoration}
    \x_{t-1} = \sqrt{\bar{\alpha}_{t-1}} {\x}_{0|t} - \mu_t \nabla_{\x}  \ell( {\x}_{0|t} ;\y) + \hat{\sigma}_t \bepsilon_\theta (\x_t;t) + \tilde{\sigma}_t \bepsilon_t,
\end{equation}
where $\mu_t$ is the guidance scaling factor. % 

Note the similarity between the PnP restoration \eqref{eq:ista_step2}-\eqref{eq:ista_step1} 
and the DDM restoration \eqref{eq:x0_est}-\eqref{eq:ddim_step_restoration}. 
Namely, ignoring the scaling of the signal by $\{ \bar{\alpha}_t \}$,
the main difference is the noise injection, composed of the estimated noise $\bepsilon_\theta$ and the random noise $\bepsilon_t$. 
The similarity of DDMs schemes to previous iterative denoising methods is further seen by the formulation of \citep{karras2022elucidating}, which decouples % 
the diffusion ``forward" flows, used for training the denoisers, from the ``reverse" % 
inference flows.
More discussion on such connections can be found in a recent review paper \citep{elad2023image}.

For example, the DDM restoration method in \citep{chung2022diffusion}, dubbed DPS, where the guidance is based on the gradient of the LS term, resembles \eqref{eq:ls_grad_step} with noise injection. 
Interestingly, just like its PnP counterpart, it requires a very large number of iterations (and thus many NFEs).
Alternatively, by using BP guidance in \eqref{eq:ddim_step_restoration} 
--- equivalently, when injecting noise to step \eqref{eq:bp_grad_step} in IDBP % 
--- one gets the % 
recent
DDNM method \citep{wang2022zero} (which is also similar to the method in \citep{song2022pseudoinverse}).
However, note that the DDNM paper focuses on noiseless settings and has difficulties in noisy settings (in Section~\ref{sec:exp} we state its technical limitation), which is aligned with the sensitivity of BP guidance to noise in $\y$, as mentioned above.

At this point, we would also like to mention the DDRM method \citep{kawar2022denoising}, 
which 
resembles the BP approach as it generates
``spectral space" measurements $\bar{\y}$ in a way similar to 
applying $\A^\dagger$ on $\y$. Yet, it
is more robust to observation noise as it requires access to the full SVD of $\A$, which allows it to mitigate noise amplification per singular component in each iteration.
In this paper we aim to devise an approach that does not require computing and storing the SVD of $\A$, which is % 
impractical is most cases.

Finally, we note that the way we construct a sampling scheme from a deterministic algorithm shares similarity with the recent method DIffPIR \citep{zhu2023denoising}. 
However, while \citep{zhu2023denoising} uses an existing PnP baseline, here we propose a novel core reconstruction method. Also, empirically our new sampling scheme demonstrates better accuracy than DiffPIR (evaluated by PSNR) without compromising on perceptual quality (evaluated by LPIPS).

\section{The Proposed Method}
\label{sec:method}
In this section, we present our image restoration method that is more robust to observation noise than methods that use only back-projections % 
while still having low computational complexity and clear benefits over methods that are based only on plain LS based guidance.
To this end, we present a novel core guidance approach, equip it with formal theoretical motivation, and utilize it to devise a new sampling-based reconstruction scheme.

{\em All the claims in this section are proved in the appendix.}

\subsection{Core approach}
\label{sec:method_core}

\tomt{
The idea for our approach comes from 
identifying \eqref{eq:bp_grad_step} as a specific instance of the formula 
\begin{equation}
\label{eq:quasi-precond}
\x_{t-1} =  {\x}_{0|t} - \mu_t \A^T \W(\A{\x}_{0|t}-\y)
\end{equation}
with $\W = (\A\A^T)^{-1}$ (assuming full rank $\A$) and $\mu_t=1$.
The plain LS step \eqref{eq:ls_grad_step} is obviously also an instance of % 
this expression 
with $\W = \I_m$.}

\tomt{
Recall traditional preconditioning in optimization \citep{nocedal1999numerical}:
The optimization
 of an objective, e.g., $L:\mathbb{R}^n\to\mathbb{R}$ in \eqref{eq:traditional_cost}, is based on 
 gradient-based optimizers where
 the full gradient $\nabla L(\x)$ is multiplied by 
\tomtb{a (symmetric) positive definite} 
matrix $\P \in \mathbb{R}^{n \times n}$. 
Such practice does not affect the problem's minimizers (since $\P\nabla L(\x)=\0 \iff \nabla L(\x)=\0$ due to the invertibility of $\P$), and its goal is to ease the optimization landscape, \tomtb{characterized by the condition number of the effective Hessian $\P^{1/2}\nabla^2 L(\x)\P^{1/2}$}. 
In \eqref{eq:quasi-precond}, however, $\W\in\mathbb{R}^{m \times m}$ does not directly affect the signal's prior and for general $\W$ we may 
not have $\A^T\W(\A\x-\y)=\0 \iff \nabla \ell_{LS}(\x)=\A^T(\A\x-\y)=\0$.}

Nevertheless, we have the following claims. % 

\begin{claim}
    Let $\A \in \mathbb{A}^{m \times n}$ with $m \leq n$ and $\mathrm{rank}(\A)=m$.
    Let $\W \in \mathbb{R}^{m \times m}$ % 
    such that $\mathrm{rank}(\A^T\W)=m$.
    Then, % 
    $$
    \A^T\W(\A\x-\y)=\0 \iff \A^T(\A\x-\y)=\0.
    $$
\end{claim}

\begin{claim}
    Let $\A \in \mathbb{A}^{m \times n}$ with $m \leq n$. 
    Let $\W \in \mathbb{R}^{m \times m}$ be a positive definite matrix that shares eigenbasis with $\A\A^T$. Then, there exists a positive definite $\P\in\mathbb{R}^{n \times n}$ such that 
    $$
    \A^T\W\A = \P^{1/2}\A^T\A\P^{1/2}.
    $$
\end{claim}

Therefore, if $\A$ is full rank, then any $\W$ for which $\A^T\W$ is full rank (e.g., $\W=(\A\A^T)^{-1}$) can be interpreted as a preconditioner of $\nabla \ell_{LS}$ (rather than of the full objective that may include additional terms).
Alternatively, even if $\A$ is not full rank, but $\W$ shares eigenbasis with $\A\A^T$ (e.g., $\W=(\A\A^T+\eta\I_m)^{-1}$), 
we allow ourselves to refer to $\W$ as a preconditioner of $\nabla \ell_{LS}$ since it has the same effect on the effective Hessian, $\P^{1/2}\nabla^2 \ell_{LS}(\x)\P^{1/2}$, as a traditional preconditioner $\P$.

Inspired by this, in lieu of \eqref{eq:ista_step1} we propose a data-fidelity guidance based on {\em iteration-dependent} preconditioner:
\begin{align}
\label{eq:irls_step}
\x_{t-1} &=  {\x}_{0|t} - \mu_t \A^T \W_t(\A{\x}_{0|t}-\y)
\end{align}
with a ``smooth" discretization of the path of symmetric preconditioners $\{\W_t\}$ --- invertible with eigenbasis shared with $\A\A^T$ --- that starts roughly at $\W_T \propto (\A\A^T+\eta\I_m)^{-1}$ and ends roughly at $\W_0 \propto \I_m$.

In words, the proposed guidance is based on gently traversing from the (regularized) BP step to the LS step along the restoration scheme.
Thus, it enjoys benefits of BP steps (stronger consistency with $\y$, accelerated convergence, etc.) in early iterations and the better robustness to noise of LS steps as the scheme approaches its end.

Potentially, a sequence of preconditioners $\{ \W_t \}$ might be learned in a data-driven manner under some optimality criterion. 
Yet, this is expected to come at the price of losing their generality to arbitrary data.
Thus, here we propose an engineered, but simple and effective, update rule of $\W_t$ 
that, as we will show empirically,  
already demonstrates the usefulness of our new concept.
Specifically, the % 
rule that % 
we propose
is given by the following convex combination:
\begin{align}
\label{eq:wt_rule}
    \W_t = (1-\delta_t)(\A\A^T+\eta\I_m)^{-1} + \delta_t % 
    \tomtb{c}\I_m
\end{align}
where $0 \leq \delta_T<...<\delta_0 \leq 1$ is a predefined sequence, % 
and $c$ is a positive scalar that balances the two terms.

In general, the main motivation for % 
\eqref{eq:wt_rule} 
is that it both obeys the previously stated % 
conditions % 
and also yields a data-fidelity guidance that is very simple to implement.
Indeed, % 
plugging \eqref{eq:wt_rule} in \eqref{eq:irls_step} yields
\begin{align}
\label{eq:simple_rule}
    \x_{t-1} &= {\x}_{0|t} - \mu_t\left ( (1-\delta_t)\g_{BP}({\x}_{0|t}) + \delta_t \g_{LS}({\x}_{0|t}) \right ),   \nonumber \\ 
    &=: {\x}_{0|t} - \mu_t \, \g_{\delta_t}({\x}_{0|t})
\end{align}
where
\begin{align}
    \g_{BP}({\x}_{0|t}) &= \A^T(\A\A^T+\eta\I_m)^{-1}(\A\x_{0|t} - \y), \\
    \g_{LS}({\x}_{0|t}) &= % 
    c \A^T(\A\x_{0|t} - \y),
\end{align}
can be computed efficiently for a wide range of observation models. 

In Section~\ref{sec:hyperparams} we will discuss practical implementations details for simple hyperparameter tuning (e.g., a simple way to tune $\{ \delta_t \}$).
As for setting $c$, observe that 
\eqref{eq:irls_step} is essentially a gradient step (at $\x_{0|t}$) of the weighted least squares (WLS) data-fidelity term:
\begin{align}
\label{eq:irls_term}
\ell_{WLS,t}(\x;\y) = \frac{1}{2}\|\W_t^{1/2}(\A\x-\y)\|_2^2,
\end{align}
for which we have the following property.

\begin{claim}
    Denote by $\lambda_1$ the largest singular value of $\A$. Let $c \leq 1/\lambda_1^2$. Then, the update \eqref{eq:simple_rule} with $\mu_t=1$ ensures reduction in \eqref{eq:irls_term}.
\end{claim}

Many popular degradation models obey $\lambda_1 \lesssim 1$ (e.g., in deblurring $\lambda_1=1$ and in super-resolution $\lambda_1$ is moderately lower than 1). 
Thus, following the claim, % 
\tomtm{unless stated otherwise, we use $c=1$ and $\mu_t=1$.}

To conclude this subsection, the core of the proposed algorithm is based on iteratively computing \eqref{eq:ista_step2} and \eqref{eq:simple_rule}. % 
Without noise injection, this is essentially a novel PnP scheme, which we call Iterative Denoising and Preconditioned Guidance (IDPG). 
As we will empirically show in Section~\ref{sec:exp}, while this method has good PSNR performance, turning it into a DDM-based sampling scheme leads to better recoveries, especially in terms of perceptual quality.
Before presenting the sampling scheme in Section~\ref{sec:method_sampling}, let us provide some theoretical motivation for our preconditioned guidance approach.

\subsection{Theory}

Let us state some mathematical claims that motivate the core guidance approach. % 
In this subsection we consider a simplified setting of having a fixed $0<\delta<1$, i.e., a fixed WLS data-fidelity term \eqref{eq:irls_term} with $\W$ given as in \eqref{eq:wt_rule} but with $\delta_t=\delta$. 
We examine the effects of this WLS, compared to LS and BP fidelity terms, on an estimator of $\x^*$, denoted by $\hat{\x}$, which is obtained by minimizing \eqref{eq:traditional_cost} with {\em Tikhonov  prior} $s(\x)=\frac{\beta}{2}\|\D\x\|_2^2$, where $\beta>0$ and $\D^T\D$ is invertible, so, the estimator $\hat{\x}$ per data-fidelity term is unique. % 
The MSE of $\hat{\x}$ (conditioned on $\x^*$, averaged over the noise $\e$) can be decomposed into two terms: squared bias, $b^2$ (independent of $\e$), and variance, $v$ (depends on $\sigma_e^2$), given by
\begin{align}
    b^2 &:= \|\mathbb{E}[\hat{\x}|\x^*] -\x^*\|_2^2, \\
    v &:= \mathrm{Tr}(\mathrm{Var}(\hat{\x}|\x^*)).
\end{align}

The motivation for analyzing this setting (even though in the proposed core procedure \eqref{eq:ista_step2} \& \eqref{eq:simple_rule}, we modify $\delta_t$ along the iterations and do not attempt to reach a minimizer per $\delta_t$)
is that insights on gains in bias 
and variance 
at the minimizer level hint that similar effects are possible even with limited optimizer iterations per $\delta_t$.
Also, in \citep{tirer2019back}, it was shown that insights gained 
for Tikhonov  prior
generalize empirically to more complex priors. 
Specifically, they showed 
that BP typically leads to smaller bias but higher variance % 
than LS.
In what follows,
the intuition that % 
$0<\delta<1$ 
enables trading bias and variance is made formal. % 

\begin{theorem}
\label{thm:bias_var}
    Consider the observation model \eqref{eq:obsevation_model} and estimating $\x^*$ via minimization of \eqref{eq:traditional_cost} with $s(\x)=\frac{\beta}{2}\|\D\x\|_2^2$. 
    Assume that: (a) $\A^T\A$ and $\D^T\D \succ 0$ share eigenbasis; (b) the singulars value of $\A$ are in $(0,1]$, and not all equal (common case); (c) $\eta=0$ and $c=1$. % 
    Then, $b_{\text{BP}}<b_{\text{WLS}}<b_{\text{LS}}$, and $v_{\text{LS}}<v_{\text{WLS}}<v_{\text{BP}}$.
\end{theorem}

Note that the margins in the theorem's inequalities depend on where $\delta$ is located in $(0,1)$ (e.g., $v_{\text{WLS}} \to v_{\text{LS}}$ for $\delta \to 1$). % 
This further motivates having a set $\{\delta_t\}$ instead of a fixed $\delta$ in \eqref{eq:wt_rule} to gain 
flexibility in handling the error's bias/variance. % 

As for the convergence rates of first-order optimization algorithms (which affect the number of NFEs with practical DNN-based priors), recall that it can typically be characterized by the condition number, $\kappa(\cdot)$, of the problem's Hessian, which equals the ratio between the largest and smallest eigenvalues of the Hessian --- the lower the better.
Note, though, that for $m<n$, the Hessian of each of the three data-fidelity terms (LS, BP, WLS) is rank-deficient due to the ill-posedness of the observation model.
Nevertheless, we have the following property on the component of their Hessian in the row-range of $\A$ (a subspace in $\mathbb{R}^n$).

\begin{claim}
    Assume that $\mathrm{rank}(\A)=m$, the singular values of $\A$ are not all equal, $\eta=0$, and denote by $\V\in\mathbb{R}^{n \times m}$ an orthonormal basis for the row-range of $\A$. % 
    We have that
    $$
    \kappa(\V^T \nabla_\x^2\ell_{BP} \V) <     \kappa(\V^T \nabla_\x^2\ell_{WLS} \V) <     \kappa(\V^T \nabla_\x^2\ell_{LS} \V).
    $$
\end{claim}

Namely, the flexibility in $\delta$ allows accelerating the reconstruction procedure compared to using only a LS based guidance.
\tomtr{Leveraging this result to 
establish ranking claims for specific algorithms
may follow the approaches of \citep{tirer2021convergence} (based on constraints on the null space of $\A$ that are implicitly imposed by the signal prior term).}

\begin{algorithm}[t]
    \caption{Denoising Diffusion with iterative Preconditioned Guidance (DDPG)}
    \label{alg:DDPG}
    \begin{algorithmic}[1]
        \Require $\bepsilon_\theta(\cdot,t)$ (noise estimator), $T$, $\y$, $\A$, $\{ \bar{\alpha}_t \}$, $\{\mu_t\}$, \tomtm{$c$}, $\eta$, $\gamma$, $\zeta$

	\If {$\sigma_e > 0$}
		\State 
            $\delta_t = \bar{\alpha}_t^\gamma$ and $w_t = \delta_t$
        \Else
            \State
            $\delta_t = 0$ and $w_t = 1$
	\EndIf
        
        \State
        Initialize $\x_{T} \sim \mathcal{N}(\mathbf{0}, \I_n)$ % 
        \For {$t$ from $T$ to $1$}
            \State ${\x}_{0|t} = \frac{1}{\sqrt{\bar{\alpha}_t}} \left ( \x_t - \sqrt{1-\bar{\alpha}_t}\bepsilon_\theta(\x_t,t) \right )$

            \State
            $\g_{BP} = \A^T(\A\A^T+\eta\I_m)^{-1}(\A\x_{0|t} - \y)$
            \State
            $\g_{LS} = % 
            \tomtm{c}\A^T(\A\x_{0|t} - \y)$
            
            \State 
            $\tilde{\x}_{t-1} = {\x}_{0|t} - \mu_t\left ( (1-\delta_t)\g_{BP} + \delta_t \g_{LS} \right )$
            
            \State
            $\hat{\bepsilon}_t = \frac{1}{\sqrt{1-\bar{\alpha}_t}}(\x_t - \sqrt{\bar{\alpha}_t} \tilde{\x}_{t-1} )$
            \State
            $\bepsilon_t \sim \mathcal{N}(\mathbf{0}, \I_n)$

            \State
            $\x_{t-1} = $
            \State
            $
            \sqrt{\bar{\alpha}_{t-1}} \tilde{\x}_{t-1} + \sqrt{1-\bar{\alpha}_{t-1}} \left ( \tomtm{w_t } \sqrt{1-\zeta} \hat{\bepsilon}_t + \sqrt{\zeta} \bepsilon_t  \right )$
            
        \EndFor \\
        \Return $\x_0$
    \end{algorithmic}
\end{algorithm}

\subsection{Sampling scheme}
\label{sec:method_sampling}

As previously mentioned, 
in this subsection we turn IDPG (alternating \eqref{eq:ista_step2} and \eqref{eq:simple_rule})
into a DDM-based sampling scheme that empirically
leads to better recoveries, especially in terms of perceptual quality. 
Injecting noise to the estimate in a certain iteration leads to an image that better matches the data on which the denoiser has been trained. Intuitively, this should improve the denoiser's performance and also allow it to mitigate error propagation along the iterations, by masking artifacts with noise that will be removed. Accordingly, it allows using a smaller value of the regularization parameter $\eta$ when handling noisy $\y$, which yields results with  
sharper details. % 

The sampling scheme that we propose is based on modification of the DDIM guided scheme, stated in \eqref{eq:ddim_step_restoration}. 
Namely, $\x_{0|t}$ is computed using a DDM's noise estimator as in \eqref{eq:x0_est} and the novel guidance $\g_{\delta_t}$ is being used in lieu of the abstract gradient $\nabla_\x\ell$.
However, we propose several important modifications of the additive noise terms $\hat{\sigma}_t \bepsilon_\theta (\x_t;t) + \tilde{\sigma}_t \bepsilon_t$ that appear in \eqref{eq:ddim_step_restoration}.

The first modification is inspired by \citep{zhu2023denoising}.
Note that, contrary to the unconditional image generation task, the data-fidelity guidance can significantly modify the estimated signal $\tilde{\x}_{t-1}:=\x_{0|t}-\mu_t\g_{\delta_t}$ compared to $\x_{0|t}$. 
Therefore, the {\em effective} predicted noise can be significantly different than $\bepsilon_\theta (\x_t;t)$.
Accordingly, based on the relation between $\bepsilon_\theta (\x_t;t)$ and $\x_{0|t}$ in \eqref{eq:x0_est}, the effective predicted noise takes the form of
\begin{align}
\hat{\bepsilon}_t = \frac{1}{\sqrt{1-\bar{\alpha}_t}}(\x_t - \sqrt{\bar{\alpha}_t} \tilde{\x}_{t-1} ),
\end{align}
and we use it instead of $\bepsilon_\theta (\x_t;t)$ for generating $\x_{t-1}$.

Similarly to this prior work,
we also simplify the way of trading between the (effective) estimated noise and the stochastic noise. Instead of time-depended noise level $\tilde{\sigma}_t$ which affects $\hat{\sigma}_t=\sqrt{1-\bar{\alpha}_{t-1}-\tilde{\sigma}_t^2}$ % 
in a complicated way, 
we balance between the levels of $\hat{\bepsilon}_t$ and $\bepsilon_t$ using the simple to tune weight $\sqrt{1-\zeta}$ and $\sqrt{\zeta}$, respectively, with a single hyperparameter $\zeta \in [0,1]$.

Our key novel modification is the following.
In the case of observation noise (i.e., $\sigma_e>0$), 
\tomtm{the effective predicted noise $\hat{\bepsilon}_t$}  
may not be useful  
in iterations where \tomtm{the overall guidance} $\g_{\delta_t}$ is dominated by $\g_{BP}$,
due to the sensitivity to noise of the BP guidance.
Therefore, we scale it 
by $\delta_t$ --- the weight that is given to $\g_{LS}$, which is more robust to observation noise. 

We remark that our sampling scheme transfers smoothly to the noiseless case, where one can simply use pure BP guidance -- or equivalently $\delta_t=0$:
it is only required to omit the last modification (i.e., avoid attenuating the noise injection by $\delta_t$).
As will be shown empirically, 
this scheme performs on par with BP-based DDM methods that are specifically devised for noiseless observations, and outperforms other methods.

The resulted % 
DDM-based sampling
algorithm, that we name Denoising Diffusion with iterative Preconditioned Guidance (DDPG), is summarized in Algorithm~\ref{alg:DDPG}.
The algorithm includes a simple way to set $\{ \delta_t \}$ that is explained in the next subsection.

\subsection{Implementation details}
\label{sec:hyperparams}

In the experimental section, we wish to test the approach (both DDPG and its core version IDPG) with powerful DDM's denoisers trained and used in  \citep{lugmayr2022repaint,guidedDiff}. 
These models are based on DDPM formulation \citep{ho2020denoising,song2020denoising}, that % 
we use to reduce
the hyperparameter tuning effort in all our examined tasks, and in particular setting $\{ \delta_t \}$ for noisy observations. (Recall that in the noiseless case, we simply set $\delta_t=0$).

As presented in Section~\ref{sec:background}, in DDPM there % 
are the parameters $\{ \bar{\alpha}_t \}$, 
which are determined by the hyperparameters $\{\beta_t\}$, and determine many other parameters (e.g., the model's noise levels $\{ \sigma_t \}$). 
These $\{ \bar{\alpha}_t \}$ obey
$0_+ \approx \bar{\alpha}_T<...<\bar{\alpha}_0=1$. 
We do not change 
$\{\beta_t\}$ and $\{ \bar{\alpha}_t \}$
at all % 
compared to previous methods that we compare with \citep{wang2022zero,kawar2022denoising}.
In fact, we use $\{ \bar{\alpha}_t \}$ to set
$\delta_t = \bar{\alpha}_t^\gamma$ 
with a single positive scalar hyperparameter $\gamma$.
This ensures that as $t$ decreases, $\{ \delta_t\}$ is monotonically increasing from 0 to 1 in rate that can be controlled by $\gamma$.
Observe that smaller $\gamma$ yields higher intermediate $\bar{\alpha}_t^\gamma$ and thus a larger portion of $\{\W_t\}$ closer to identity matrix --- or, equivalently, ``faster" progress from BP to LS. This is beneficial for larger observation noise $\sigma_e$. 

To conclude, % 
\tomtm{when using $c=1$ and $\mu_t=1$,} 
applying our approach requires tuning only $\eta$, $\gamma$ and $\zeta$ (for DDPG).
Note that, potentially, tuning more hyperparameters for our method can further boost its results.

\section{Experiments}
\label{sec:exp}

In this section, 
we examine the performance of our approach.
We start with % 
examining 
our core guidance, as reflected in IDPG, compared to its extreme cases: deterministic methods with pure BP and LS based guidance.
Then, we compare both IDPG and DDPG to existing methods: DDRM \cite{kawar2022denoising}, \tomtr{DPS} \cite{chung2022diffusion}, DDNM \cite{wang2022zero}, and DiffPIR \cite{zhu2023denoising}.
We consider the CelebA-HQ 1K and ImageNet 1K test sets, where for fair comparison all the methods use the same pretrained DDM's denoisers: the model from \citep{lugmayr2022repaint} trained on CelebA-HQ and the (unconditional) model from \cite{guidedDiff} trained on ImageNet.
We use $T=100$ iterations for each of the methods as in \citep{wang2022zero}.
\tomtr{An exception is DPS, which requires $T=1000$ iterations with much larger per-iteration complexity, making it extremely slow compared to others.\footnote{Using the same hardware, DPS takes a dozen minutes per image compared to a couples of seconds of the other methods, and limiting it to 100 iterations fails to produce meaningful recoveries.}}
All the methods are initialized with the same random $\x_T$.

We consider image deblurring and super-resolution tasks that have been widely examined in the previous works, and for each method we use the hyperparameter settings suggested by its authors. The settings for our methods are stated in the appendix. 
Specifically, we consider super-resolution with bicubic downsampling and scale factor 4 and deblurring with Gaussian kernel as, e.g., in \citep{kawar2022denoising,wang2022zero}.
However, we also examine challenging noisy deblurring with random motion kernels, as in \cite{zhu2023denoising}.
This is an example for a common task that cannot be handled by DDRM \cite{kawar2022denoising}, which is limited to separable kernels for which the SVD of $\A$ can be efficiently computed and stored.
This task cannot be addressed also by DDNM \cite{wang2022zero}, which fails to tackle noisy observations.\footnote{Note that DDNM is tailored for $\sigma_e=0$. 
We remark that DDNM+ that was proposed in \citep{wang2022zero} for handling noisy $\y$ (via SVD) seems to be 
heavily tied to a specific downsampling task (without bicubic kernel) and does not support the considered tasks, \tomtf{as shown in the 
appendix.} 
Our efforts to adapt it failed.} 
We examine the noisy cases also in the two aforementioned tasks.

\tomtr{In the appendix, 
we demonstrate the advantages of our approach over recent task-specific DNNs,  
\tomtf{as well as its superior performance also with lower noise level than demonstrated here.} 
In addition, we report there the impressive performance of our approach for sparse-view computed tomography (another task where computing the SVD of $\A$ is infeasible).}

\begin{table}
\scriptsize % 
\renewcommand{\arraystretch}{1.3}
\caption{\small Data-fidelity guidance ablation: super-resolution and deblurring PSNR [dB] ($\uparrow$) and LPIPS ($\downarrow$) results on CelebA-HQ 1K.} \label{table:results_celeba_ablation}
\centering
    \begin{tabular}{ | c || c | c | c | c | c | c | c |}
    \hline
 Task  & PGM-LS & IDBP & IDPG (ours)  \\ \hline \hline

    Bicub.~SRx4~$\sigma_e$=0 &  32.40 / 0.127 & {\bf 32.66} / {\bf 0.111} & {\bf 32.66} / {\bf 0.111}  \\ \hline

    Bicub.~SRx4~$\sigma_e$=0.05 &  28.91 / 0.209 & 28.95 / 0.229 & {\bf 29.89} / {\bf 0.155}  \\ \hline

    Gauss.~Deb.~$\sigma_e$=0 &  32.25 / 0.141 & {\bf 45.58} / {\bf 0.002} & {\bf 45.58} / {\bf 0.002}  \\ \hline
    
    Gauss.~Deb.~$\sigma_e$=0.05 &   28.61 / 0.191 & 30.55 / {\bf 0.150} & {\bf 31.08} / {\bf 0.150}  \\ \hline
    
    Gauss.~Deb.~$\sigma_e$=0.1 & 27.89 / 0.239 & 28.29 / 0.226 & {\bf 29.28} / {\bf 0.146}  \\ \hline    
    
    \end{tabular}
\end{table}

\begin{figure}
    \center
    \begin{subfigure}[h]{0.25\columnwidth}
        \centering
        \includegraphics[width=\textwidth]{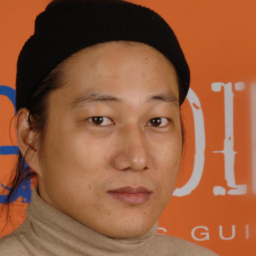}
    \end{subfigure} % 
    \begin{subfigure}[h]{0.25\columnwidth}
        \centering
        \includegraphics[width=\textwidth]{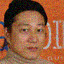}
    \end{subfigure} % 
    \begin{subfigure}[h]{0.25\columnwidth}
        \centering
        \includegraphics[width=\textwidth]{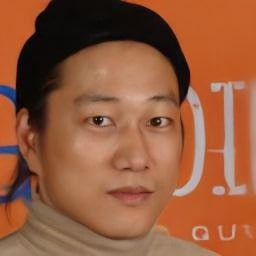}
    \end{subfigure} \hspace{-1mm} \\
    
    \begin{subfigure}[h]{0.25\columnwidth}
        \centering
        \includegraphics[width=\textwidth]{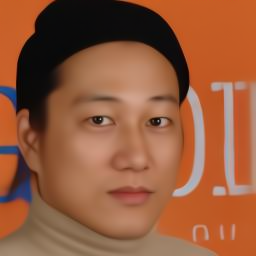}
    \end{subfigure}  % 
    \begin{subfigure}[h]{0.25\columnwidth}
        \centering
        \includegraphics[width=\textwidth]{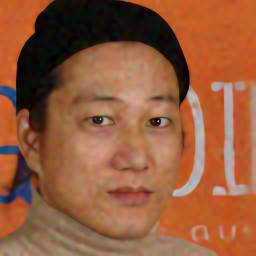}
    \end{subfigure}  % 
    \begin{subfigure}[h]{0.25\columnwidth}
        \centering
        \includegraphics[width=\textwidth]{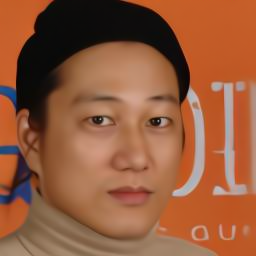}
    \end{subfigure}
    \vspace{-1mm}
    \caption{
    SRx4 with noise level 0.05. Top row, from left to right: original, upsampled observation, and our DDPG.
    Bottom row, from left to right: PGM-LS, IDBP, and our IDPG (baseline for DDPG).}
    \label{fig:SRx4_0.05_core}
    \vspace{-1mm}
\end{figure}

\begin{table*}
\scriptsize % 
\renewcommand{\arraystretch}{1.3}
\caption{Super-resolution and deblurring PSNR [dB] ($\uparrow$) and LPIPS ($\downarrow$) results on CelebA-HQ 1K. N/A marks applicability limitation of: (1) DDNM to noiseless settings and (2) DDRM to settings where the SVD is given and stored. (More details in the text).} % 
\label{table:results_celeba}
\vspace{-1mm}
\centering
    \begin{tabular}{ | c || c | c | c | c | c | c | c |}
    \hline
 \diagbox[height=2em,width=10em]{Task}{Method}  & DDRM & DPS (1000 NFEs) 
 & DiffPIR & DDNM & 
 IDPG (ours) & DDPG (ours) \\ \hline \hline

    Bicub.~SRx4~$\sigma_e$=0 &  31.64 / 0.054  
    & 29.39 / 0.065
    & 30.26 / 0.051 & 31.64 / {\bf 0.048} & 
    {\bf 32.66} / 0.111  & 31.60 / 0.052 \\ \hline

    Bicub.~SRx4~$\sigma_e$=0.05 &  29.26 / 0.090 
    & 27.49 / 0.086 
    & 27.44 / {\bf 0.085} & N/A & 
    {\bf 29.89} / 0.155 & 29.39 / 0.105  \\ \hline

    Gauss.~Deb.~$\sigma_e$=0 & 42.49 / 0.006 
    & 31.25 / 0.055
    & 32.97 / 0.041 & 45.56 / {\bf 0.002} & 
    {\bf 45.58} / {\bf 0.002} & 45.46 / {\bf 0.002}  \\ \hline
    
    Gauss.~Deb.~$\sigma_e$=0.05 &  30.53 /  0.074 
    & 27.75 / 0.084
    & 28.89 / 0.074 & N/A & 
    {\bf 31.08} / 0.150 & 30.41 / {\bf 0.068}  \\ \hline
    Gauss.~Deb.~$\sigma_e$=0.1 &  28.79 / 0.088 
    & 26.67 / 0.097
    & 27.59 / 0.083 & N/A & 
    {\bf 29.28} / 0.146 & 29.18 / {\bf 0.080}  \\ \hline

    Motion Deb.~$\sigma_e$=0.05 &  N/A 
    & 19.63 / 0.227
    & 27.96 / 0.102 & N/A & 
     {\bf 29.73} / 0.134  & 29.02 / {\bf 0.082}\\ \hline
    Motion Deb.~$\sigma_e$=0.1 &  N/A 
    & 19.64 / 0.231
    & 26.23 / 0.132 & N/A & 
     {\bf 27.86} / 0.166 & 27.74 / {\bf 0.099}\\ \hline    
    \end{tabular}
    \vspace{-1mm}

\vspace{5mm}

\scriptsize % 
\renewcommand{\arraystretch}{1.3}
\caption{\small Super-resolution and deblurring PSNR [dB] ($\uparrow$) and LPIPS ($\downarrow$) results on Imagenet 1K. N/A marks applicability limitation of: (1) DDNM to noiseless settings and (2) DDRM to settings where the SVD is given and stored. (More details in the text).} % 
\label{table:results_imagenet}
\vspace{-1mm}
\centering
    \begin{tabular}{ | c || c | c | c | c | c | c | c |}
    \hline
 \diagbox[height=2em,width=10em]{Task}{Method}  & DDRM & DPS (1000 NFEs) 
 & DiffPIR & DDNM & % 
 IDPG (ours) & DDPG (ours) \\ \hline \hline

    Bicub.~SRx4~$\sigma_e$=0 &  % 
    27.38 / 0.270% 
    & 25.56 / 0.236
    & 26.99 / {\bf 0.225} & {\bf 27.45} / 0.245 % 
    & 27.20 / 0.326  & 27.41 / 0.255 \\ \hline

    Bicub.~SRx4~$\sigma_e$=0.05 &   25.54 / 0.333 % 
    & 24.05 / {\bf 0.271}
    & 24.65 /  0.318 & N/A % 
    & 25.51 / 0.411 & {\bf 25.55} / 0.354  \\ \hline

    Gauss.~Deb.~$\sigma_e$=0 & 40.53 / 0.008 % 
    & 25.54 / 0.259
    & 30.54 / 0.082 & 43.64 / 0.003 % 
    & 44.02 / {\bf 0.002} & {\bf 44.21} / {\bf 0.002}  \\ \hline
    
    Gauss.~Deb.~$\sigma_e$=0.05 &  27.71 / 0.243 % 
    & 23.59 / 0.294
    & 26.64 / 0.240 & N/A % 
    & 27.47 / 0.313 & % 
    {\bf 27.73} / {\bf 0.205}  % 
    \\ \hline

    Motion Deb.~$\sigma_e$=0.05 &  N/A 
    & 17.52 / 0.468
    & 25.34 / 0.284 & N/A & 
     {\bf 26.02} / 0.354 &   25.94 / {\bf 0.249}  \\ \hline

    \end{tabular}
    \vspace{-1mm}
\end{table*}

\subsection{Examining the core approach}
\label{sec:exp_core}

As the main contribution of our paper is the novel guidance approach that smoothly traverses from BP to LS guidance, we start with examining our baseline IDPG method (alternating \eqref{eq:ista_step2} and \eqref{eq:simple_rule}) with
IDBP \citep{tirer2018image} ($\delta_t=0$) and a method with LS based updates ($\delta_t=1$), denoted here by PGM-LS (a common PnP scheme \cite{kamilov2023plug}). This comparison essentially provides an ablation study for our core approach.

The results for bicubic super-resolution and Gaussian deblurring with various and without noise for the CelebA-HQ 1K test set are presented in Table~\ref{table:results_celeba_ablation}. 
They show that IDPG consistently outperforms IDBP and PGM-LS. % 
Figure~\ref{fig:SRx4_0.05_core} displays qualitative results, showing that IDPG has less artifacts than IDBP and finer details than PGM-LS.
Yet, while IDPG provides high PSNR values its perceptual quality can be improved by using the sampling-based DDPG, as shown in this figure.

\subsection{Comparison with other methods}
\label{sec:exp_comparison}

We turn to compare the proposed approach to DDRM, \tomtr{DPS}, DDNM and DiffPIR.
We consider all deblurring and super-resolution tasks that have been described above with various noise levels. 
The results for the CelebA-HQ 1K test set and the ImageNet 1K test set are presented in Table~\ref{table:results_celeba} and Table~\ref{table:results_imagenet}, respectively. 
We present qualitative results for  
motion deblurring, Gaussian deblurring and super-resolution in Figures~\ref{fig:motion_deblur_0.05}, \ref{fig:g_deb_0.05_imagenet} and \ref{fig:SRx4_0.05_celeba}, respectively.  
More results appear in the appendix.

Examining the results, 
observe that, indeed, the sampling-based version of our approach (DDPG) consistently leads to better perceptual quality, both visually and as measured by LPIPS.
Interestingly, while there is an inherent tradeoff between accuracy and perceptual quality \cite{blau2018perception}, the PSNR reduction in DDPG compared to IDPG is rather moderate (and in some case DDPG even has slightly better PSNR than IDPG).
One can observe that our approach typically outperforms other approaches in deblurring and is competitive in super-resolution.

Importantly, note that % 
the only reference methods that are as flexible as our approach to the observation model \tomtr{are DPS, which is very slow, and DiffPIR}.
However, their PSNR is significantly lower than the PSNR of our DDPG.
Recall that in many critical applications (e.g., medical imaging), accuracy, as measured by the classical PSNR metric, is extremely important.

\begin{figure}
    \center
    \begin{subfigure}[h]{0.25\columnwidth}
        \centering
        \includegraphics[width=\textwidth]{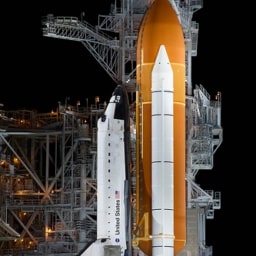}
    \end{subfigure} % 
    \begin{subfigure}[h]{0.25\columnwidth}
        \centering
        \includegraphics[width=\textwidth]{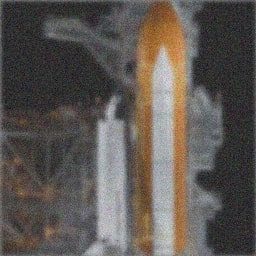}
    \end{subfigure} % 
    \begin{subfigure}[h]{0.25\columnwidth}
        \centering
        \includegraphics[width=\textwidth]{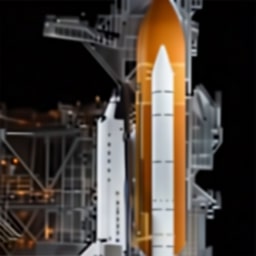}
    \end{subfigure} \hspace{-1mm} \\ 
    
    \begin{subfigure}[h]{0.25\columnwidth}
        \centering
        \includegraphics[width=\textwidth]{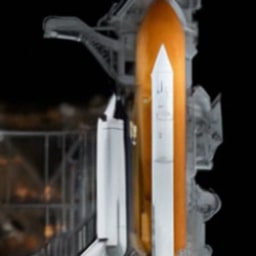}
    \end{subfigure}  % 
    \begin{subfigure}[h]{0.25\columnwidth}
        \centering
        \includegraphics[width=\textwidth]{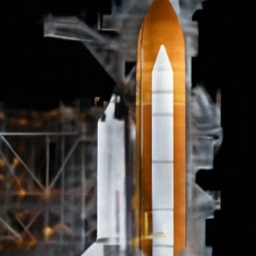}
    \end{subfigure}  % 
    \begin{subfigure}[h]{0.25\columnwidth}
        \centering
        \includegraphics[width=\textwidth]{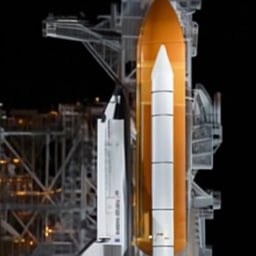}
    \end{subfigure}
    \vspace{-2mm}
    \caption{
    Gaussian deblurring with noise level 0.05. Top row, from left to right: original, observation, and our IDPG (baseline for DDPG). 
    Bottom row, from left to right: DPS, DiffPIR, and our DDPG.} % 
    \label{fig:g_deb_0.05_imagenet}
    \vspace{-1mm}
    \center
    \begin{subfigure}[h]{0.25\columnwidth}
        \centering
        \includegraphics[width=\textwidth]{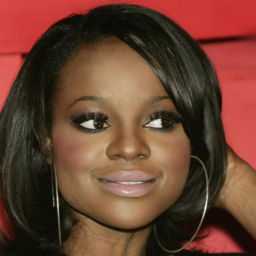}
    \end{subfigure} % 
    \begin{subfigure}[h]{0.25\columnwidth}
        \centering
        \includegraphics[width=\textwidth]{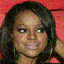}
    \end{subfigure} % 
    \begin{subfigure}[h]{0.25\columnwidth}
        \centering
        \includegraphics[width=\textwidth]{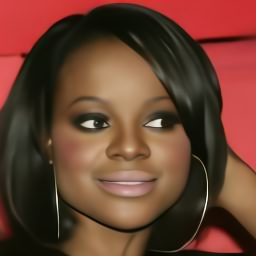}
    \end{subfigure} \hspace{-1mm} \\
    
    \begin{subfigure}[h]{0.25\columnwidth}
        \centering
        \includegraphics[width=\textwidth]{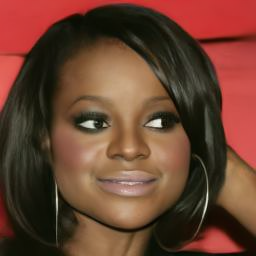}
    \end{subfigure}  % 
    \begin{subfigure}[h]{0.25\columnwidth}
        \centering
        \includegraphics[width=\textwidth]{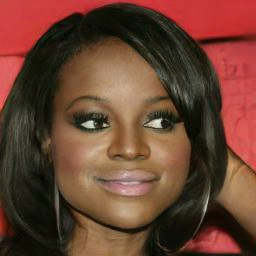}
    \end{subfigure}  % 
    \begin{subfigure}[h]{0.25\columnwidth}
        \centering
        \includegraphics[width=\textwidth]{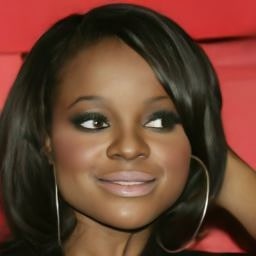}
    \end{subfigure}
    \vspace{-2mm}
    \caption{
    SRx4 with noise level 0.05. Top row, from left to right: original, upsampled observation, and our IDPG (baseline for DDPG). 
    Bottom row, from left to right: DDRM, DiffPIR, and our DDPG.}
    \label{fig:SRx4_0.05_celeba}
    \vspace{-4mm}
\end{figure}

\section{Conclusion}
\label{sec:conclusion}

In this paper,
we presented a framework for solving linear inverse problems with DNN denoisers/diffusers, which uses a novel preconditioned data-fidelity guidance approach, based on traversing from back-projection steps to least squares steps, exploiting the advantages of each.
We used the new approach with a computationally convenient trajectory 
within a ``plug-and-play" optimization scheme % 
and 
its DDM sampling-based counterpart, \tomtr{which we devised}. 
The performance advantages of the new technique were shown in various deblurring and super-resolution settings, 
with and without observation noise.
\tomtm{(See the appendix 
for computed tomography experiments as well).}
As a direction for future research, one can try to learn the preconditioners $\{\W_t\}$ instead of designing them.

\vspace{1mm}
\noindent
\textbf{Acknowledgment.} 
The work is 
partially 
supported by ISF grant no.~1940/23.

\newpage

{
    \small
    \bibliographystyle{ieeenat_fullname}
    \bibliography{main_cvpr_final}
}

\newpage
\onecolumn
\appendix

\section{Proofs}
\label{sec:proofs}

\begin{claim}
    Let $\A \in \mathbb{A}^{m \times n}$ with $m \leq n$ and $\mathrm{rank}(\A)=m$.
    Let $\W \in \mathbb{R}^{m \times m}$ % 
    such that $\mathrm{rank}(\A^T\W)=m$.
    Then, we have
    $$
    \A^T\W(\A\x-\y)=\0 \iff \A^T(\A\x-\y)=\0.
    $$
\end{claim}

\begin{proof}
    Since $\mathrm{rank}(\A)=m$ we have that $\A^T(\A\x-\y)=\0 \iff \A\x-\y=\0$ (e.g., multiply both sides of $\A^T(\A\x-\y)=\0$ from left by $\A^{T\dagger}=(\A\A^T)^{-1}\A$). 
    Similarly, since $\mathrm{rank}(\A^T\W)=m$ we have that $\A^T\W(\A\x-\y)=\0 \iff \A\x-\y=\0$ (e.g., multiply both sides from left by $(\A^T\W)^\dagger$). Thus we get the required result. 
\end{proof}

\begin{claim}
    Let $\A \in \mathbb{A}^{m \times n}$ with $m \leq n$. 
    Let $\W \in \mathbb{R}^{m \times m}$ be a positive definite matrix that shares eigenbasis with $\A\A^T$. Then, there exists a positive definite $\P\in\mathbb{R}^{n \times n}$ such that 
    $$
    \A^T\W\A = \P^{1/2}\A^T\A\P^{1/2}.
    $$
\end{claim}

\begin{proof}
    Let $\A=\U\bLambda\V^T$ be the SVD of $\A$, where $\bLambda\in\mathbb{R}^{m \times n}$ is rectangular diagonal, and $\U\in\mathbb{R}^{m \times m}$ and $\V\in\mathbb{R}^{n \times n}$ are orthogonal matrices. By the assumptions on $\W$ we have $\W=\U\bGamma\U^T$, where $\bGamma$ is diagonal and invertible. Thus, we have
    $$
    \A^T\W\A = \V\bLambda^T\bGamma\bLambda\V^T.
    $$
    Pick $\P = \V\tilde{\bGamma}\V^T$ where $\tilde{\bGamma} \in \mathbb{R}^{n \times n}$ is a diagonal matrix with the first $m$ entries on its diagonal that are the same as $\bGamma$ and 1's (or any other positive values) in the lower $n-m$ entries. We have
    $$
    \P^{1/2}\A^T\A\P^{1/2} = \V \tilde{\bGamma}^{1/2} \bLambda^T \bLambda \tilde{\bGamma}^{1/2} \V^T = \V\bLambda^T\bGamma\bLambda\V^T,
    $$
    which concludes the proof.
\end{proof}

\begin{claim}
    Denote by $\lambda_1$ the largest singular value of $\A$. Let $c \leq 1/\lambda_1^2$. Then, the update \eqref{eq:simple_rule} with $\mu_t=1$ ensures reduction in \eqref{eq:irls_term}.
\end{claim}

\begin{proof}

We begin by showing that under such choice of $c$, we have that $\g_{\delta_t}(\cdot)=\nabla_{\x} \ell_{WLS,t}(\cdot;\y)$ is $1$-Lipschitz. We prove it by upper bounding the operator norm of the Hessian $\nabla_{\x}^2\ell_{WLS,t}(\cdot;\y)$ by 1:
\begin{align}
    \|\nabla_{\x}^2\ell_{\W_t}\| &= \|\A^T\W_t\A\|  \\ \nonumber
    &\leq (1-\delta_t)\| \A^T(\A\A^T+\eta\I_m)^{-1}\A \| + \delta_t c \| \A^T\A \| \\ \nonumber
    &\leq (1-\delta_t) + \delta_t = 1.
\end{align}
where in the first inequality follows from the triangle inequality and the second inequality follows from $\| \A^T(\A\A^T+\eta\I_m)^{-1}\A \| \leq 1$ and $\|\A^T\A\|=1/\lambda_1^2$.

The claim is a consequence of the descent lemma for the gradient step $\tilde{\x}=\x-\mu_t \nabla_{\x} \ell_{WLS,t}(\x;\y)$ when the step-size equals 1 over the Lipschitz constant of $\g_{\delta_t}=\nabla_{\x} \ell_{WLS,t}$, which is 1 in our case.

For completeness, we present this well-known result here.
To simplify notation we denote $\ell_{WLS,t}$ by $\ell$ and omit dependency on $\y$. 
The 1-Lipschitzness of the gradient implies that
$\| \nabla_{\x} \ell(\x_2) - \nabla \ell(\x_1) \|_2 \leq \| \x_2 - \x_1 \|_2$ for all $\x_2, \x_1$. Equivalently,  % 
this implies that for all $\x_2, \x_1$ we have 
\begin{align}
\label{Eq_smoothness_cond}
\ell(\x_2) - \ell(\x_1) \leq \nabla \ell(\x_1)^T (\x_2 - \x_1) + \frac{1}{2} \| \x_2 - \x_1 \|_2^2.
\end{align}
Recall that $\tilde{\x}=\x-\mu_t \nabla \ell(\x)$, so using 
$\x_1={\x}$ and $\x_2=\tilde{\x}$ in \eqref{Eq_smoothness_cond}, we get
\begin{align}
\label{Eq_gd_iter2}
\ell(\tilde{\x}) - \ell({\x}) \leq -\mu_t \| \nabla \ell({\x}) \|_2^2 + \mu_t^2 \frac{1}{2} \| \nabla \ell({\x}) \|_2^2.
\end{align}
Finally, substituting $\mu_t=1$ gives $\ell(\tilde{\x}) - \ell({\x}) \leq -\frac{1}{2}\| \nabla \ell({\x}) \|_2^2 \implies \ell(\tilde{\x}) < \ell({\x})$ whenever $\nabla \ell({\x}) \neq \0$.

\end{proof}

\begin{theorem}
\label{thm:bias_var_app}
    Consider the observation model \eqref{eq:obsevation_model} and estimating $\x^*$ via minimization of \eqref{eq:traditional_cost} with $s(\x)=\frac{\beta}{2}\|\D\x\|_2^2$. 
    Assume that: (a) $\A^T\A$ and $\D^T\D \succ 0$ share eigenbasis; (b) the singulars value of $\A$ are in $(0,1]$, and not all equal (common case); (c) $\eta=0$ and $c=1$. % 
    Then, $b_{\text{BP}}<b_{\text{WLS}}<b_{\text{LS}}$, and $v_{\text{LS}}<v_{\text{WLS}}<v_{\text{BP}}$.
\end{theorem}

\begin{proof}

Let us define the singular value decomposition (SVD) of the $m \times n$ matrix $\A=\U \bLambda \V^T$, where $\U$ is an $m \times m$ orthogonal matrix whose columns are the left singular vectors, $\bLambda$ is an $m \times n$ rectangular diagonal matrix with nonzero singular values $\{\lambda_i \}_{i=1}^{m}$ on the diagonal, and $\V$ is an $n \times n$ orthogonal matrix whose columns are the right singular vectors. 
The assumptions on $\D$ imply that $\D^T\D=\V \bGamma^2 \V^T \succ 0$, where $\bGamma^2$ is an $n \times n$ diagonal matrix of nonzero eigenvalues $\{\gamma_i^2 \}_{i=1}^{n}$.

Recall that we consider the cost function
$$
f_{WLS}(\x) = \frac{1}{2}\|\W^{1/2}(\A\x-\y)\|_2^2 + \frac{\beta}{2} \|\D\x\|_2^2.
$$
Due to the (strong) convexity of the cost function, the (unique) minimizer can be obtained simply by equating their gradients to zero
\begin{align}
\label{eq:cost_typical_est}
\nabla f_{WLS}(\tilde{\x}) &= \A^T \W ( \A {\x} - \y ) + \beta \D^T\D {\x} = \0  \nonumber \\
& \Rightarrow \hat{\x}_{WLS}=(\A^T\W\A + \beta \D^T\D)^{-1}\A^T \W \y.
\end{align}
Note that $\hat{\x}_{LS}$ and $\hat{\x}_{BP}$ are instances of this formula with $\W=\I_m$ and $\W=(\A\A^T)^{-1}$, respectively.
For the WLS under consideration we have $\W = (1-\delta)(\A\A^T)^{-1} + \delta\I_m$.
In all these cases we have $\W=\U\S\U^T$ where $\S$ is an $m \times m$ diagonal matrix of positive values $\{ s_i \}_{i=1}^m$ (eigenvalues of $\W$).

From the conditions of the noise we have that $\mathbb{E}[\e]=\0$ and $\mathbb{E}[\e\e^T]=\sigma_e^2\I_m$.
Thus, similarly to the analysis in \citep{tirer2019back}, the MSE (conditioned on $\x^*$) can be expressed as
\begin{align}
\label{Eq_cost_typical_mse}
\Exp_\e\|\hat{\x}_{WLS} - \x^* \|_2^2 
& = \Exp_\e \left \|(\A^T\W\A + \beta \D^T\D)^{-1}\A^T\W (\A\x^*+\e) - \x^* \right \|_2^2  \nonumber \\
& = \tomt{ \left \|  (\A^T\W\A + \beta \D^T\D)^{-1}\A^T\W\A \x^* - \x^* \right \|_2^2 } + \Exp_\e \left [ \e^T \W\A (\A^T\W\A + \beta \D^T\D)^{-2} \A^T \W \e \right ]   \nonumber\\
&= \tomt{ \left \|\left ( (\A^T\W\A + \beta \D^T\D)^{-1}\A^T\W\A - \I_n \right )\x^* \right \|_2^2}  + \sigma_e^2 \mathrm{Tr} \left ( (\A^T\W\A + \beta \D^T\D)^{-2} \A^T \W^2 \A \right )  \nonumber\\
& = \left \| \V \left ( (\bLambda^T\S\bLambda + \beta \bGamma^2)^{-1}\bLambda^T\S\bLambda - \I_n \right ) \V^T \x^* \right \|_2^2 + \sigma_e^2 \mathrm{Tr} \left ( \V (\bLambda^T\S\bLambda + \beta \bGamma^2)^{-2} \bLambda^T\S^2\bLambda \V^T \right ) \nonumber \\
&= \sum \limits_{i=1}^{n} \Big ( \frac{\lambda_i^2 s_i}{\lambda_i^2 s_i + \beta \gamma_i^2} - 1 \Big )^2 [\V^T\x]_i^2 + \sigma_e^2 \sum \limits_{i=1}^{n} \frac{\lambda_i^2 s_i^2}{(\lambda_i^2 s_i + \beta \gamma_i^2)^2}
\end{align}
where $s_i$ and $\lambda_i$ with $i>m$ are just used for notation convenience and are in fact zeros.

The first term in \eqref{Eq_cost_typical_mse} is the squared bias and the second term is the variance. 
These expressions can be specialized to each data-fidelity term by substituting the relevant $\S$.
Specifically, we have that the bias terms of the estimators are given by:
\begin{align}
\label{Eq_cost_all_bias}
bias_{LS}^2 &=  \sum \limits_{i=1}^{m} % 
\Big ( \frac{\beta \gamma_i^2}{\lambda_i^2+\beta \gamma_i^2} \Big )^2 [\V^T\x^*]_i^2 % 
+  \sum \limits_{i=m+1}^{n} [\V^T\x^*]_i^2 , \nonumber \\
bias_{BP}^2 &=  \sum \limits_{i=1}^{m} % 
\Big ( \frac{\beta \gamma_i^2}{1+\beta \gamma_i^2} \Big )^2 [\V^T\x^*]_i^2 % 
+  \sum \limits_{i=m+1}^{n} [\V^T\x^*]_i^2 , \nonumber \\
bias_{WLS}^2 &= 
\sum \limits_{i=1}^{m} % 
\Big ( \frac{\beta \gamma_i^2}{(1-\delta) + \delta\lambda_i^2 + \beta \gamma_i^2} \Big )^2 [\V^T\x^*]_i^2 
+  \sum \limits_{i=m+1}^{n} [\V^T\x^*]_i^2,
\end{align}
where we used the fact that for $\W = (1-\delta)(\A\A^T)^{-1} + \delta\I_m$ we have that $s_i = (1-\delta)/\lambda_i^2+\delta$.

By the theorem's assumption $\lambda_i \in (0,1]$ and not all are equal.
Thus, 
we have that $\lambda_i^2 \leq (1-\delta)+\delta\lambda_i^2 \leq 1$
with strict inequalities at some $i$. 
Therefore, to prove $bias_{BP}^2 < bias_{WLS}^2 < bias_{LS}^2$, it suffices to show that
the function $f(x)=\left( \frac{a}{x+a} \right)^2$ with $a>0$ is strictly monotonic decreasing on $[0,1]$, and this trivially holds.

Let us now consider the variances:
\begin{align}
\label{Eq_cost_all_var}
var_{LS} &= \sigma_e^2 \sum \limits_{i=1}^{m} % 
\lambda_i^{-2} \frac{\lambda_i^4}{\left ( \lambda_i^2+\beta \gamma_i^2 \right )^2} ,
\nonumber \\
var_{BP} &= \sigma_e^2 \sum \limits_{i=1}^{m} % 
\lambda_i^{-2} \frac{ 1}{(1+\beta \gamma_i^2)^2} 
, \nonumber \\
var_{WLS} &= \sigma_e^2 \sum \limits_{i=1}^{m} % 
\frac{ \lambda_i^2((1-\delta)/\lambda_i^2 + \delta)^2}{\left (\lambda_i^2((1-\delta)/\lambda_i^2 + \delta) + \beta \gamma_i^2 \right )^2} 
= \sigma_e^2 \sum \limits_{i=1}^{m} \lambda_i^{-2} \frac{ ((1-\delta) + \delta \lambda_i^2 )^2}{\left ((1-\delta) + \delta \lambda_i^2 + \beta \gamma_i^2 \right )^2}
\end{align}

Similarly to the way the bias terms where compared, since $\lambda_i^2 \leq (1-\delta)+\delta\lambda_i^2 \leq 1$
with strict inequalities at some, to prove $var_{BP}^2 > var_{WLS}^2 > var_{LS}^2$, it suffices to show that
the function $f(x)= \frac{x^2}{(x+a)^2} = \frac{1}{(1+a/x)^2}$ with $a>0$ is strictly monotonic increasing on $(0,1]$, and this trivially holds.

\end{proof}

\begin{claim}
    Assume that $\mathrm{rank}(\A)=m$, the singular values of $\A$ are not all equal, $\eta=0$, and denote by $\V\in\mathbb{R}^{n \times m}$ an orthonormal basis for the row-range of $\A$. % 
    We have that
    $$
    \kappa(\V^T \nabla_\x^2\ell_{BP} \V) <     \kappa(\V^T \nabla_\x^2\ell_{WLS} \V) <     \kappa(\V^T \nabla_\x^2\ell_{LS} \V).
    $$
\end{claim}

\begin{proof}

We can write the {\em compact} SVD of $\A$ as
$\A=\U\bLambda\V^T$, where $\bLambda\in\mathbb{R}^{m \times m}$ is  a diagonal matrix with nonzero
singular values $\{\lambda_i\}_{i=1}^m$ (indexed in decreasing order), $\U\in\mathbb{R}^{m \times m}$ is an orthogonal matrix and $\V\in\mathbb{R}^{n \times m}$ is the stated partial orthogonal matrix.
Note that $\nabla_\x^2\ell_{BP}=\A^T(\A\A^T)^{-1}\A$, $\nabla_\x^2\ell_{LS}=\A^T\A$, and $\nabla_\x^2\ell_{WLS}=(1-\delta)\A^T(\A\A^T)^{-1}\A + \delta c \A^T\A$.
Thus, $\kappa(\V^T\nabla_\x^2\ell_{BP}\V)=\kappa(\I_m)=1$.
$\kappa(\V^T\nabla_\x^2\ell_{LS}\V)=\kappa(\bLambda^2)=\frac{\lambda_1^2}{\lambda_m^2}$.
$\kappa(\V^T\nabla_\x^2\ell_{WLS}\V)=\kappa((1-\delta)\I_m+\delta c \bLambda^2)=\frac{(1-\delta)+\delta c \lambda_1^2}{(1-\delta)+\delta c \lambda_m^2}$.
Clearly, $\kappa(\V^T\nabla_\x^2\ell_{WLS}\V) > \kappa(\V^T\nabla_\x^2\ell_{BP}\V)$.
And $\kappa(\V^T\nabla_\x^2\ell_{WLS}\V) < \kappa(\V^T\nabla_\x^2\ell_{LS}\V)$ follows from
$$
\frac{(1-\delta)+\delta c \lambda_1^2}{(1-\delta)+\delta c \lambda_m^2} < \frac{\lambda_1^2}{\lambda_m^2}
\iff
\lambda_m^2\left ( (1-\delta)+\delta c \lambda_1^2 \right ) < \lambda_1^2 \left ( (1-\delta)+\delta c \lambda_m^2 \right ) \iff \lambda_m^2 < \lambda_1^2.
$$
    
\end{proof}

\newpage
\section{Fast Pseudoinverse Implementations}
\label{sec:pseudoinverse}

In this section, we show that the pseudoinverse operation $\A^\dagger:\mathbb{R}^m \to \mathbb{R}^n$ can be implemented very efficiently for the cases of image deblurring and image super-resolution (no need to compute and store the SVD of $\A$). We note that there are other cases where this operation can be easily implemented, such as image inpainting, computed tomography, and more.
In image inpainting we simply have that $\A^\dagger=\A^T$. In fact, this is the case whenever $\A$ is a {\em tight-frame} (i.e., when $\A\A^T=\I_m$). In this case, the BP and LS update steps are essentially equivalent, and therefore do not require the special treatment that is considered in the paper. In computed tomography, the pseudoinverse can be implemented via fast 
\tomtr{(filtered) inverse}
Radon transform, whose details are out of the scope of this paper.
Moreover, as mentioned in the paper, for general $\A$ one can implement the operation  $\A^\dagger=\A^T(\A\A^T)^\dagger$ with low computational complexity by the conjugate gradients methods, where full rank $\A\A^T$ (and $\A\A^T+\eta\I_m$ otherwise) can be ``inverted" using few conjugate gradient iterations, which only require applying the operations $\A$ and $\A^T$ and bypass the need of matrix inversion or SVD.

\subsection{Image Deblurring}

In image deblurring the measurement operator $\A \in \mathbb{R}^{n \times n}$ (note that $m=n$) is a convolution with some blur kernel $\mathbf{k}$, i.e., $\A\x = \x \circledast \mathbf{k}$. 
Under the assumption of circular convolution (which merely affects boundary pixels and can be addressed by padding), we have that $\A$ is a circulant matrix, and thus can be diagonalized by the discrete Fourier transform.
Therefore, this convolution operation can be computed as element-wise multiplication in the discrete Fourier domain, which is efficiently implemented via Fast Fourier Transform (FFT). 
Specifically, for $\z\in\mathbb{R}^n$ we have that $\A\z = \mathcal{F}^{-1}\left (\mathcal{F}(\mathbf{k})\mathcal{F}(\z) \right)$, where $\mathcal{F}$ denotes the FFT. 
Similarly, $\A^T$, which is convolution with flipped $\mathbf{k}$, can be applied as $\A^T\z = \mathcal{F}^{-1}\left (\overline{\mathcal{F}(\mathbf{k})}\mathcal{F}(\z) \right)$.
Lastly, the operation $\A^T(\A\A^T+\eta\I_n)^{-1}\z$ can be efficiently computed as
\begin{align}
    \A^T(\A\A^T+\eta\I_n)^{-1}\z = \mathcal{F}^{-1} \left ( \frac{ \overline{\mathcal{F}(\mathbf{k})}\mathcal{F}(\z)}{|\mathcal{F}(\mathbf{k})|^2 + \eta} \right ).
\end{align}
As done throughout the paper, we use notation of 1D signal vector for simplification, but the extension to 2D signals, 2D convolutions, and 2D FFT, is straightforward.

\subsection{Image Super-Resolution}

In image super-resolution the measurement operator $\A \in \mathbb{R}^{m \times n}$ (note that $m=n$) is a composition of convolution with some blur kernel $\mathbf{k}$ and subsampling by some scale factor $s$, i.e., $\A\x = [\x \circledast \mathbf{k}]\downarrow_s$. 

Under the assumption of circular convolution (which merely affects boundary pixels and can be addressed by padding) and integer $s=n/m$, we have $\A=\S\B$, where $\B\in\mathbb{R}^{n \times n}$ is a circulant matrix and $\S \in \mathbb{R}^{m \times n}$.
Therefore, the operation $\A=\S\B$ can be implemented by FFT-based filtering followed by subsampling and the operation $\A^T=\B^T\S^T$ can be implemented by upsampling followed by FFT-based filtering.
Moreover, $\A\A^T=\S\B\B^T\S^T$ is circulant and essentially performs filtering with the kernel $\mathbf{k}_0 = \left [ \mathcal{F}^{-1}\left (|\mathcal{F}(\mathbf{k})|^2 \right) \right ] \downarrow_s$.
Lastly, the operation $\A^T(\A\A^T+\eta\I_m)^{-1}\z$ can be efficiently computed as
\begin{align}
    \A^T(\A\A^T+\eta\I_m)^{-1}\z = \mathcal{F}^{-1} \left ( \overline{\mathcal{F}(\mathbf{k})} \mathcal{F} \left ( \left [ \mathcal{F}^{-1} \left ( \frac{ \mathcal{F}(\z)}{|\mathcal{F}(\mathbf{k}_0)|^2 + \eta} \right ) \right ] \big \uparrow_s \right ) \right ).
\end{align}
Again, extension from 1D to 2D is straightforward.

\newpage

\section{More Experimental Details and Results}

In this section we present more details on the experiments, and more quantitative and qualitative results, which have not been stated in the main body of the paper due to space limitation. % 
Our code is available at
\url{https://github.com/tirer-lab/DDPG}.

\begin{table}[H]
\scriptsize % 
\renewcommand{\arraystretch}{1.3}
\caption{Super-resolution and deblurring PSNR [dB] ($\uparrow$) and LPIPS ($\downarrow$) results on CelebA-HQ 1K. N/A marks applicability limitation of: (1) DDNM to noiseless settings and (2) DDRM to settings where the SVD is given and stored. (More details in the text).
Note that SwinIR and Restormer are task-specific methods, and are thus not flexible to handle most of the examined tasks. 
} \label{table:results_celeba_supp}
\centering
    \begin{tabular}{ | c || c | c | c | c | c | c | c | c | c |}
    \hline
 \diagbox[height=2em,width=10em]{Task}{Method}  & SwinIR (SR) & Restormer (Deb.) % 
 & DDRM & DPS (1000 NFEs) 
 & DiffPIR & DDNM & 
 IDPG (ours) & DDPG (ours) \\ \hline \hline

    Bicub.~SRx4~$\sigma_e$=0 
    & {\bf 33.26} / 0.100 & ---
    & 31.64 / 0.054  
    & 29.39 / 0.065
    & 30.26 / 0.051 & 31.64 / {\bf 0.048} & 
    32.66 / 0.111  & 31.60 / 0.052 \\ \hline

    Bicub.~SRx4~$\sigma_e$=0.05 
    & 27.30 / 0.213 & ---
    & 29.26 / 0.090 
    & 27.49 / 0.086
    & 27.44 / {\bf 0.085} & N/A & 
    {\bf 29.89} / 0.155 & 29.39 / 0.105  \\ \hline

    Gauss.~Deb.~$\sigma_e$=0 
    & --- & 29.32 / 0.100
    & 42.49 / 0.006 
    & 31.25 / 0.055
    & 32.97 / 0.041 & 45.56 / {\bf 0.002} & 
    {\bf 45.58} / {\bf 0.002} & 45.46 / {\bf 0.002}  \\ \hline
    
    Gauss.~Deb.~$\sigma_e$=0.05 
    & --- & 25.28 / 0.431
    & 30.53 / 0.074 
    & 27.75 / 0.084
    & 28.89 / 0.074 & N/A & 
     {\bf 31.08} / 0.150 & 30.41 / {\bf 0.068}  \\ \hline
    
    Gauss.~Deb.~$\sigma_e$=0.1 
    & --- & 21.67 / 0.652
    & 28.79 / 0.088 
    & 26.67 / 0.097
    & 27.59 / 0.083 & N/A & 
    {\bf 29.28} / 0.146 & 29.18 / {\bf 0.080}  \\ \hline

    Motion Deb.~$\sigma_e$=0.05 
    & --- & 19.03 / 0.530
    &  N/A 
    & 19.63 / 0.227
    & 27.96 / 0.102 & N/A & 
    {\bf 29.73} / 0.134 & 29.02 / {\bf 0.082}  \\ \hline
    
    Motion Deb.~$\sigma_e$=0.1 
    & --- & 16.32 / 0.813
    &  N/A 
    & 19.64 / 0.231
    & 26.23 / 0.132 & N/A & 
    {\bf 27.86} / 0.166 & 27.74 / {\bf 0.099}  \\ \hline    
    
    \end{tabular}
\end{table}

\begin{wrapfigure}{r}{0.25\textwidth}
\vspace{-30pt}
    \centering
  \includegraphics[width=100pt]{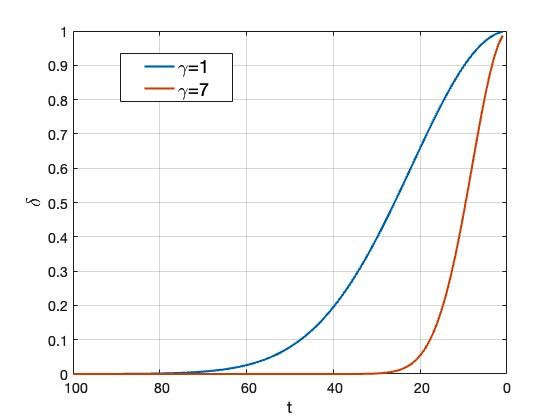}  
  \vspace{-3mm}
  \caption{\footnotesize $\delta_t$ for $\gamma=\{1,7\}$.}
\label{fig:delta_ddpg}
\vspace{-10pt}
\end{wrapfigure}

\subsection{Hyperparameter setting}

As mentioned in Section~\ref{sec:hyperparams}, in our experiments we do not modify the denoising diffusion model (DDM) hyperparameters $\{ \beta_t \}$ compared to other methods. Specifically, we have that this set is composed of linear scheduling from  $\beta_{start}=0.0001$ to $\beta_{end}=0.02$. % 
The parameters $\{ \bar{\alpha}_t \}$ are determined by $\{ \beta_t \}$. 
As explained in the paper, we use $\{ \bar{\alpha}_t \}$ of size $T=100$ to set $\{ \delta_t \}$ via $\delta_t = \bar{\alpha}_t^\gamma$, where $\gamma \geq 0$ is a single hyperparameter that we tune. 
Figure~\ref{fig:delta_ddpg} shows the resulting $\{ \delta_t \}$ for two values of $\gamma$.
Note that if $\sigma_e=0$ we simply set $\delta_t=0$, so we do not need to tune $\gamma$.

Another hyperparameter is $\eta$, which regularizes the inversion in the operation $\A^T(\A\A^T+\eta\I_m)^{-1}$.
We scale it according to the noise level and define: $\eta=\mathrm{max}(1\mathrm{e}{-4}, (2\sigma_e)^2 \tilde{\eta})$, where $\tilde{\eta}$ is the hyperparameter that we tune. 
Note that if $\sigma_e=0$ we do not need to tune $\tilde{\eta}$. 
\tomtm{Setting $c=1$, it is left} 
to set the step-size $\{ \mu_t \}$ and, specifically for DDPG, also $\zeta \in [0,1]$.  
The step-size that is used is either $\mu_t=1$ or $\mu_t = (1-\bar{\alpha}_{t-1})/(1-\bar{\alpha}_{t}):=\mu_t^*$, which reduces from 1 significantly only close to the last iterations. 

As mentioned in Section~\ref{sec:exp}, the tasks that we consider are common in the literature.
For super-resolution, we consider bicubic downsampling with scale factor 4, as in \citep{kawar2022denoising,wang2022zero}. For deblurring, we consider Gaussian blur kernel with standard deviation 10 clipped to size $5 \times 5$, as in \citep{kawar2022denoising,wang2022zero}.
For deblurring, we also consider motion blur kernels generated using the same procedure (with intensity value 0.5) as in \cite{chung2022diffusion,zhu2023denoising}.  
For each observation model we consider different levels of Gaussian noise out of $\{0, 0.05, 0.1\}$.

Let us state the hyperparameters for Section~\ref{sec:exp_core} (examining the core approach).
IDBP is tuned with $\tilde{\eta}=\{32,6\}$ for deblurring and SR, respectively. 
For $\sigma_e=0$, IDPG reduces to IDBP ($\delta_t=0$) \textcolor{black}{and $\tilde{\eta}$ is irrelevant.} 
\textcolor{black}{Regarding CelebA-HQ with noisy observations,} 
\tomtf{for SR with $\sigma_e=0.05$ it is used with $\tilde{\eta}=0.2$ and $\gamma=16$, and for Gaussian deblurring it is used with $\tilde{\eta}=0.6$ and $\gamma=\{8,6\}$ for $\sigma_e=\{0.05,0.1\}$, respectively.}
Additionally, for motion deblurring in Section~\ref{sec:exp_comparison}, IDPG is tuned with $\gamma=\{12,14\}$ and $\tilde{\eta}=\{0.9, 1\}$ (in this case, larger $\tilde{\eta}$ for larger noise allows increasing $\gamma$).
\textcolor{black}{Regarding ImageNet with $\sigma_e=0.05$, we use $\gamma=\{30,11,14\}$ and $\tilde{\eta}=\{0.2, 0.6, 0.8\}$ for SR, Gaussian deblurring and motion deblurring, respectively.} % 
\textcolor{black}{In all these cases, we use IDPG with $\mu_t=1$.}

Lastly, the hyperparameters of DDPG are listed in Table~\ref{table:ddpg_hyperparams}.
\textcolor{black}{Note that in many of the settings,  $\tilde{\eta}$ and $\zeta$ remain similar.}

\begin{table}[H]
\scriptsize % 
\renewcommand{\arraystretch}{1.3}
\caption{DDPG hyperparameters.} 
\vspace{-2mm}
\label{table:ddpg_hyperparams}
\centering
    \begin{tabular}{ | c || c | c | c | c | c | c | c |}
    \hline
 Task  & CelebA-HQ & ImageNet \\ \hline \hline

    Bicub.~SRx4~$\sigma_e$=0 
    & $\zeta=0.7$, $\mu_t = 1$ 
    & $\zeta=0.7$, $\mu_t = 1$ \\ \hline

    Bicub.~SRx4~$\sigma_e$=0.05 
    &  $\gamma=10.0$, $\zeta=0.8$, $\tilde{\eta}=0.3$, $\mu_t=\mu_t^*$ 
    & $\gamma=6.0$, $\zeta=1.0$, $\tilde{\eta}=0.3$, $\mu_t=\mu_t^*$  \\ \hline

    Gauss.~Deb.~$\sigma_e$=0 
    & $\zeta=1.0$, $\mu_t = 1$   
    & $\zeta=1.0$, $\mu_t = 1$  \\ \hline
    
    Gauss.~Deb.~$\sigma_e$=0.05 
    & $\gamma=8.0$, $\zeta=0.5$, $\tilde{\eta}=0.7$  , $\mu_t = \mu_t^*$  
    & $\gamma=10.0$, $\zeta=0.4$, $\tilde{\eta}=0.7$  , $\mu_t = \mu_t^*$   \\ \hline
    
    Gauss.~Deb.~$\sigma_e$=0.1 
    & $\gamma=5.0$, $\zeta=0.6$, $\tilde{\eta}=0.7$ , $\mu_t = \mu_t^*$  
    & --- \\ \hline

    Motion Deb.~$\sigma_e$=0.05 
    &   $\gamma=5.0$, $\zeta=0.6$, $\tilde{\eta}=0.6$ , $\mu_t = \mu_t^*$
    & % 
    
    $\gamma=6.0$, $\zeta=0.6$, $\tilde{\eta}=0.7$, $\mu_t = \mu_t^*$ \\ \hline
    
    Motion Deb.~$\sigma_e$=0.1 
    &   $\gamma=5.0$, $\zeta=0.6$, $\tilde{\eta}=0.6$, $\mu_t = \mu_t^*$
    &  $\gamma=3.0$, $\zeta=0.6$, $\tilde{\eta}=0.4$, $\mu_t = \mu_t^*$
    \\ \hline   
    
    \end{tabular}
\end{table}

\subsection{More quantitative comparisons for deblurring and super-resolution}

\tomtr{In this subsection, we report results of more competing methods for the same experimental settings that appear in the main body of the paper.}

We examine 
two representative deep learning methods that are based on per-task supervised learning: 
SwinIR \cite{liang2021swinir} for super-resolution and Restormer \cite{zamir2022restormer} for deblurring.
Note though that, as discussed in the paper, we observed that these methods do not generalize well to test sets that are not exactly aligned with their exhaustive training procedure.
Specifically, while SwinIR performs well (in terms of PSNR but not in terms of LPIPS) for the noiseless SRx4 with bicubic downsampling, for which it has been exactly trained, it exhibits massive performance drop in the presence of noise. 
Similarly, we could not managed to get good results with the Restormer, presumably because its training phase considered a specific deblurring dataset.
In fact, the behavior of these methods motivates using deep learning for learning the signal prior separately from the observation model, as we discussed in the introduction section.

The results for CelebA-HQ 1K test set are presented in Table~\ref{table:results_celeba_supp} (which is an extended version of Table~\ref{table:results_celeba}).
The discussion on the results, as made in the main body of the paper, still carries on.
Both our IDPG and DDPG are flexible to the observation model. IDPG presents good PSNR results and DDPG balances it with good LPIPS results (and better perceptual quality). In general, our DDPG demonstrates competitive LPIPS results and better PSNR results than the alternative DDM-based methods. The only reference methods that are as flexible to the observation model as DDPG are DiffPIR \cite{zhu2023denoising} and DPS \cite{chung2022diffusion}. However, DiffPIR yields significantly lower PSNR and DPS both yields lower PSNR and is also extremely slow.

\begin{wrapfigure}{r}{0.2\columnwidth}
\vspace{-12pt}
    \centering
  \includegraphics[width=50pt]{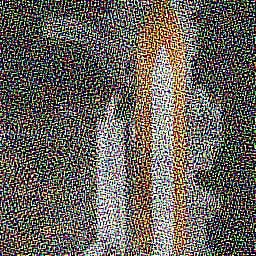}  
  \vspace{-3mm}
  \caption{\scriptsize % 
  Failure of DDNM+ for Gaussian deblurring with noise level 0.05.} % 
\label{fig:g_deb_0.05_imagenet_small}
\vspace{-10pt}
\end{wrapfigure}

\vspace{3mm}

\noindent
\tomtf{
\textbf{Applicability issues of DDNM+.}
As mentioned in Section~\ref{sec:exp}, DDNM+ that was proposed in \cite{wang2022zero} for handling noisy $\y$, {\em via SVD} (!), seems to be heavily tied to a specific downsampling task (without bicubic kernel) and does not support the considered tasks. % 
Indeed, when running the official DDNM+ code for bicubic SR with noise we get ``not supported" assert, and when running it for deblurring Gaussian kernel with noise level 0.05 (as in Figure~\ref{fig:g_deb_0.05_imagenet}) % 
it completely fails,
e.g., see 
Figure~\ref{fig:g_deb_0.05_imagenet_small}. Thus, DDNM+ cannot be applied to the examined settings (and all the efforts to fix it failed).} % 

\vspace{3mm}

\noindent
\tomtf{
\textbf{Low-noise scenarios.} 
Note that the fact that our approach handles well both noiseless settings and settings with high noise levels implies that it can be readily used for settings with low noise levels.
In Table~\ref{table:results_celeba_tiny} we present the results for $\sigma_e=0.01$, which 
show the advantages of our approach also 
in low noise scenarios.
In all these cases we use $\mu_t=1$. For SR we use $\gamma=300, \zeta=1.0, \tilde{\eta}=1.0$.
For Gaussian deblurring we use $\gamma=11, \zeta=0.6, \tilde{\eta}=1.0$. For motion deblurring we use $\gamma=50, \zeta=0.5, \tilde{\eta}=6.0$.}

\begin{table}[H]
\vspace{-2mm}
\scriptsize % 
\renewcommand{\arraystretch}{1.3}
\caption{% 
PSNR and LPIPS for CelebA-HQ 1K with $\sigma_e=0.01$. DDRM is not applicable for motion deblur. DDNM(+) is not applicable.}
\label{table:results_celeba_tiny}
\vspace{-2mm}
\centering
    \begin{tabular}{ | c || c | c | c | c | c | c | c |}
    \hline
 \diagbox[height=2em,width=10em]{Task}{Method}  & DDRM & DPS % 
 & DiffPIR & 
 IDPG (ours) & DDPG (ours) \\ \hline \hline

    Bicub.~SRx4~$\sigma_e$=0.01 & 31.09 / 0.066  
    & 29.11 / 0.068 
    & 29.62 / {\bf 0.058} & 
    {\bf 31.99} / 0.127 & 31.81 / 0.092  \\ \hline
    
    Gauss.~Deb.~$\sigma_e$=0.01 &  33.90 / 0.045 
    & 30.27 / 0.060 
    & 32.01 / 0.060 & 
    {\bf 34.26} / 0.071  & 32.20 / {\bf 0.044}  \\ \hline

    Motion Deb.~$\sigma_e$=0.01 &  N/A 
    & 19.52 / 0.228
    & 31.72 / 0.050 & 
    {\bf 33.29} / 0.079 &  32.55 / {\bf 0.045}  \\ \hline    

    \end{tabular}
    \vspace{-6mm}
\end{table}

\begin{figure}
    \centering
    \begin{subfigure}[h]{0.16\textwidth}
        \caption*{Ground truth}
        \centering
        \includegraphics[width=\textwidth]{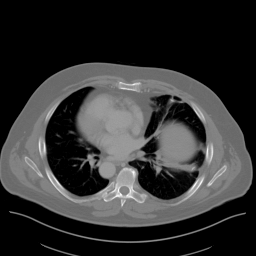}
    \end{subfigure}
    \begin{subfigure}[h]{0.16\textwidth}
        \caption*{FBP}
        \centering
        \includegraphics[width=\textwidth]{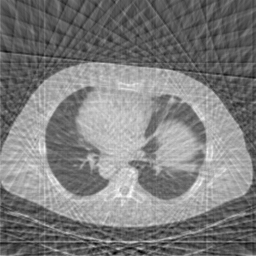}
    \end{subfigure}
    \begin{subfigure}[h]{0.16\textwidth}
        \caption*{MCG}
        \centering
        \includegraphics[width=\textwidth]{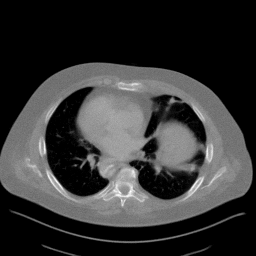}
    \end{subfigure}
    \begin{subfigure}[h]{0.16\textwidth}
        \caption*{{\bf DDPG}}
        \centering
        \includegraphics[width=\textwidth]{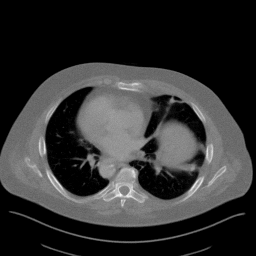}
    \end{subfigure}
    \\
    \begin{subfigure}[h]{0.16\textwidth}
        \caption*{Ground truth}
        \centering
        \includegraphics[width=\textwidth]{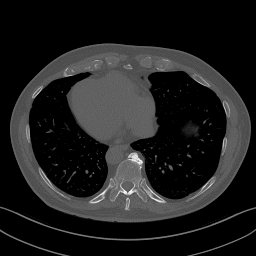}
    \end{subfigure}    
    \begin{subfigure}[h]{0.16\textwidth}
        \caption*{FBP}
        \centering
        \includegraphics[width=\textwidth]{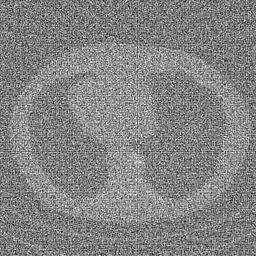}
    \end{subfigure}
    \begin{subfigure}[h]{0.16\textwidth}
        \caption*{MCG}
        \centering
        \includegraphics[width=\textwidth]{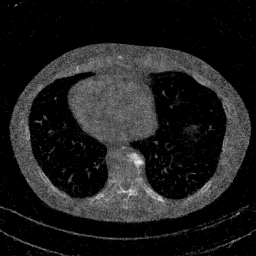}
    \end{subfigure}
    \begin{subfigure}[h]{0.16\textwidth}
        \caption*{{\bf DDPG}}
        \centering
        \includegraphics[width=\textwidth]{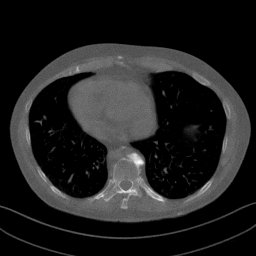}
    \end{subfigure}
    \caption{AAPM: Sparse-view CT (30 views). Top: $\sigma_e=0$; bottom: $\sigma_e=0.001\|\A\x^*\|_2$.}
    \label{fig:CT}
\end{figure}

\subsection{\tomtr{Sparse-view computed tomography}}

In this subsection, we report the performance of our DDPG for sparse-view computed tomography (SV-CT).
We compare our method against the recent MCG method \cite{chung2022improving}, which has an official implementation for such task, based on score-SDE model \cite{song2020score}, pre-trained on the 2016 American Association of Physicists in Medicine (AAPM) grand challenge dataset resized to $256 \times 256$ resolution. 
As done in \cite{chung2022improving}, the measurement operator $\A$
simulates the CT measurement process with
parallel beam geometry with evenly-spaced 180 degrees (essentially, implemented by applying Radon transform on $\x^*$). 
The test set consists of 100 
held-out validation images from the AAPM challenge.

To demonstrate the ease of integrating our approach in SDE-based sampling schemes (and not only in DDPM/DDIM schemes), we make minimal modifications to the MCG implementation, and essentially, merely replace their data-fidelity guidance with our $\g_{\delta_t}$.
Specifically, we keep using $T=2000$ iterations as in MCG (though, this number can be reduced) with the same set of  noise levels $\{ \tilde{\lambda}_t \} \in (0,1]$ that decreases along the iterations.
Conveniently, we set 
the step-size $\mu_t = \tilde{\lambda}_t$, and 
$\delta_t= \left ( \frac{1-\tilde{\lambda}_t}{1-\mathrm{min}\tilde{\lambda}} \right )^\gamma$. Thus, we can still tune only a scalar $\gamma$ to determine $\{\delta_t\}$ for our DDPG.
No $\zeta$ needs to be tune, as the estimated noise in not injected (equivalently $\zeta=1$).
Regarding the regularized back-projection operation (used in $\g_{BP}$), in the context of CT, it is typically being referred to as ``filtered back-projection" (FBP) and it is implemented by incorporating Ramp filter with the inverse Radon transform. 
The Ramp filter is  triangular in frequency domain with values between 0 and 1 that attenuates low frequencies and thus emphasizes details. 
We impose the regularization on this BP operation via the hyperparameter $\eta$ simply by upper bounding the filter in frequency domain by $1/\eta$ (so, e.g., $\eta=0$ implies no regularization).
As for the LS step (used in $\g_{LS}$), the largest eigenvalue of $\A$, denoted by $\lambda_1$ in the main body of the paper, is larger than 1 for CT, so we set $c=1/\lambda_1^2$ instead of $c=1$. 
To conclude, we have only two hyperparameters, $\gamma$ and $\eta$, that we manually tune for DDPG.

\begin{wraptable}{r}{5.5cm}
\scriptsize % 
\renewcommand{\arraystretch}{1.3}
\caption{Sparse-view CT (30 views): PSNR [dB] ($\uparrow$) and SSIM ($\uparrow$) results on AAPM dataset.}
\label{table:results_CT_supp}
\vspace{-3mm}
\centering
    \begin{tabular}{ | c || c | c | c | c | c | c | c |}
    \hline
 \diagbox[height=2em,width=8em]{Task}{Method}  & MCG & DDPG (ours)  \\ \hline \hline

    CT, $\sigma_e=0$ &  34.98 / 0.905  & {\bf 36.01} / {\bf 0.913}  \\ \hline
    
    CT, $\sigma_e>0$  & 23.63 / 0.480 & {\bf 26.75} / {\bf 0.761}  \\ \hline    
    
    \end{tabular}
\end{wraptable}

We consider the SV-CT with 30 views (as in \cite{chung2022improving}). We examine the case where we do not add additional Gaussian noise $\e$ to $\A\x^*$. Yet, we observed that some ground truth images are already noisy and, presumably, this is detrimental for pure BP-based guidance. 
We also examine the case where the additional noise level is $0.001\|\A\x^*\|_2$. % 
We use $\gamma=1, \eta=0$ and $\gamma=0.1, \eta=10$ for the two cases, respectively.
The quantitative results (PSNR and SSIM metrics) are presented in Table~\ref{table:results_CT_supp}.
They show that DDPG outperforms MCG.  
Qualitative results, which are presented in Figure~\ref{fig:CT}, visually demonstrate the superiority of DDPG over MCG in recovering finer details and robustness to noise.

\subsection{More qualitative results}

In what follows, we present more visual results for the different tasks.
In the noiseless cases many of the methods perform well, so we recommend the reader to focus on the results for the noisy settings, which are also the focus of the paper.

\begin{figure}[b]
    \centering
       \begin{subfigure}[h]{0.16\textwidth}
        \caption*{Ground truth}
        \centering
        \includegraphics[width=\textwidth]{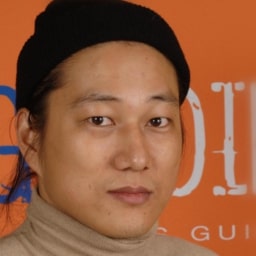}
    \end{subfigure}
    \begin{subfigure}[h]{0.16\textwidth}
        \caption*{Observed image}
        \centering
        \includegraphics[width=\textwidth]{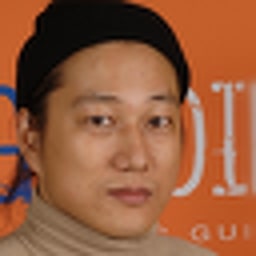}
    \end{subfigure}
    \begin{subfigure}[h]{0.16\textwidth}
        \caption*{SwinIR}
        \centering
        \includegraphics[width=\textwidth]{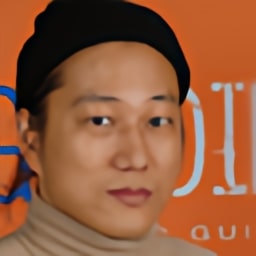}
    \end{subfigure}
    \begin{subfigure}[h]{0.16\textwidth}
        \caption*{DDRM}
        \centering
        \includegraphics[width=\textwidth]{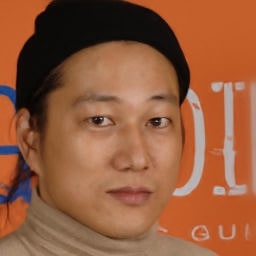}
    \end{subfigure}
    \begin{subfigure}[h]{0.16\textwidth}
        \caption*{DDNM}
        \centering
        \includegraphics[width=\textwidth]{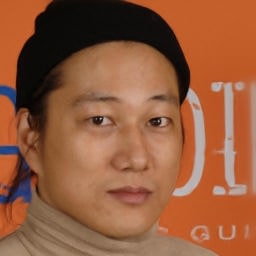}
    \end{subfigure}
    \\
    \begin{subfigure}[h]{0.16\textwidth}
        \caption*{DPS}
        \centering
        \includegraphics[width=\textwidth]{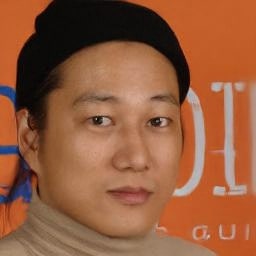}
    \end{subfigure}
    \begin{subfigure}[h]{0.16\textwidth}
        \caption*{DiffPIR}
        \centering
        \includegraphics[width=\textwidth]{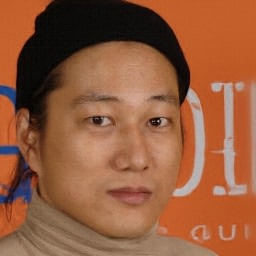}
    \end{subfigure}
    \begin{subfigure}[h]{0.16\textwidth}
        \caption*{{\bf IDPG}}
        \centering
        \includegraphics[width=\textwidth]{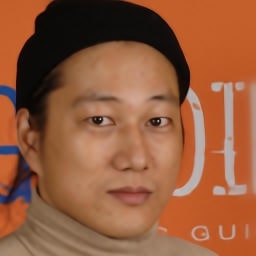}
    \end{subfigure}
    \begin{subfigure}[h]{0.16\textwidth}
        \caption*{{\bf DDPG}}
        \centering
        \includegraphics[width=\textwidth]{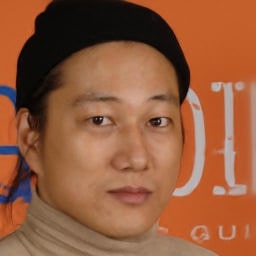}
    \end{subfigure}
    
    \caption{CelebA-HQ: SRx4 for noiseless bicubic downsampling.}
    \label{fig:SRx4_0_celeba}
\end{figure}

\begin{figure}
    \centering
        \begin{subfigure}[h]{0.16\textwidth}
        \caption*{Ground truth}
        \centering
        \includegraphics[width=\textwidth]{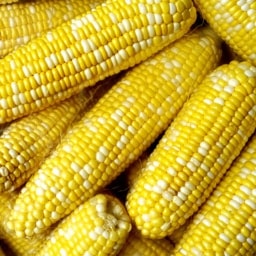}
    \end{subfigure}
    \begin{subfigure}[h]{0.16\textwidth}
        \caption*{Observed image}
        \centering
        \includegraphics[width=\textwidth]{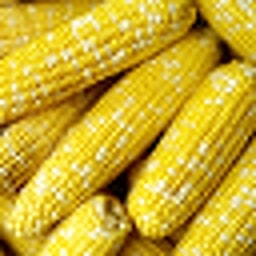}
    \end{subfigure}
    \begin{subfigure}[h]{0.16\textwidth}
        \caption*{SwinIR}
        \centering
        \includegraphics[width=\textwidth]{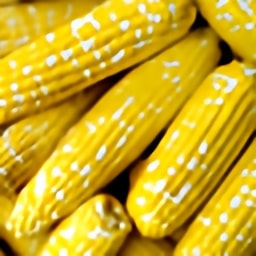}
    \end{subfigure}
    \begin{subfigure}[h]{0.16\textwidth}
        \caption*{DDRM}
        \centering
        \includegraphics[width=\textwidth]{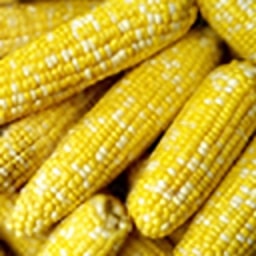}
    \end{subfigure}
    \begin{subfigure}[h]{0.16\textwidth}
        \caption*{DDNM}
        \centering
        \includegraphics[width=\textwidth]{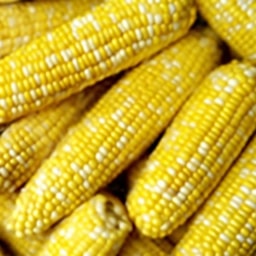}
    \end{subfigure}
    \\
    \begin{subfigure}[h]{0.16\textwidth}
        \caption*{DPS}
        \centering
        \includegraphics[width=\textwidth]{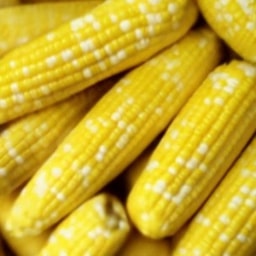}
    \end{subfigure}
    \begin{subfigure}[h]{0.16\textwidth}
        \caption*{DiffPIR}
        \centering
        \includegraphics[width=\textwidth]{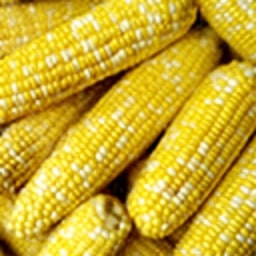}
    \end{subfigure}
    \begin{subfigure}[h]{0.16\textwidth}
        \caption*{{\bf IDPG}}
        \centering
        \includegraphics[width=\textwidth]{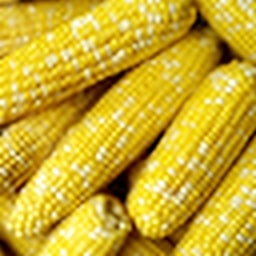}
    \end{subfigure}
    \begin{subfigure}[h]{0.16\textwidth}
        \caption*{{\bf DDPG}}
        \centering
        \includegraphics[width=\textwidth]{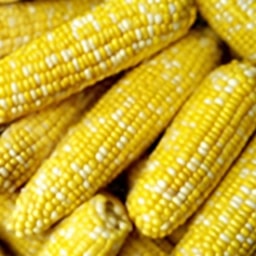}
    \end{subfigure}
    \caption{ImageNet: SRx4 for noiseless bicubic downsampling.}
    \label{fig:SRx4_0_imagenet}
\end{figure}

\begin{figure}
    \centering
    \begin{subfigure}[h]{0.16\textwidth}
        \caption*{Ground truth}
        \centering
        \includegraphics[width=\textwidth]{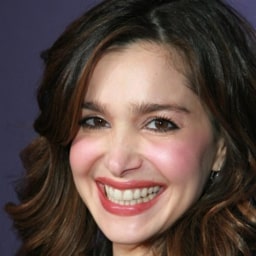}
    \end{subfigure}
    \begin{subfigure}[h]{0.16\textwidth}
        \caption*{Observed image}
        \centering
        \includegraphics[width=\textwidth]{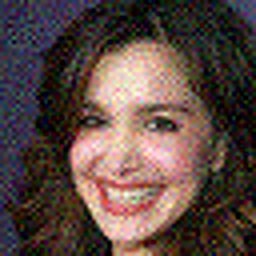}
    \end{subfigure}
    \begin{subfigure}[h]{0.16\textwidth}
        \caption*{SwinIR}
        \centering
        \includegraphics[width=\textwidth]{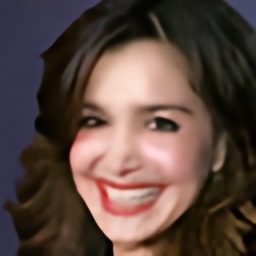}
    \end{subfigure}
    \begin{subfigure}[h]{0.16\textwidth}
        \caption*{DDRM}
        \centering
        \includegraphics[width=\textwidth]{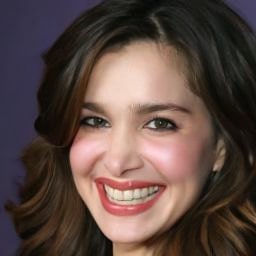}
    \end{subfigure}
    \\
    \begin{subfigure}[h]{0.16\textwidth}
        \caption*{DPS}
        \centering
        \includegraphics[width=\textwidth]{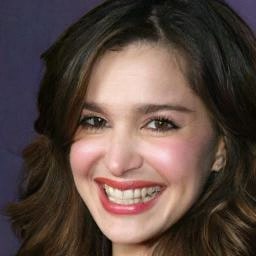}
    \end{subfigure}    
    \begin{subfigure}[h]{0.16\textwidth}
        \caption*{DiffPIR}
        \centering
        \includegraphics[width=\textwidth]{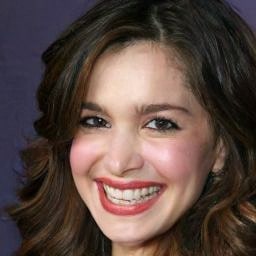}
    \end{subfigure}
    \begin{subfigure}[h]{0.16\textwidth}
        \caption*{{\bf IDPG}}
        \centering
        \includegraphics[width=\textwidth]{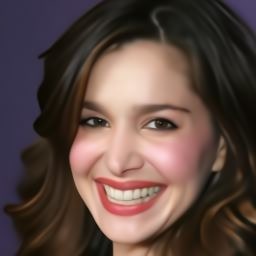}
    \end{subfigure}
    \begin{subfigure}[h]{0.16\textwidth}
        \caption*{{\bf DDPG}}
        \centering
        \includegraphics[width=\textwidth]{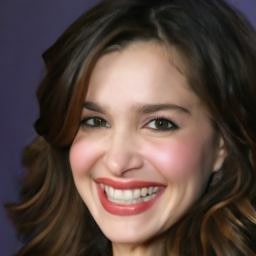}
    \end{subfigure}
    \\
    \centering
    \begin{subfigure}[h]{0.16\textwidth}
        \caption*{Ground truth}
        \centering
        \includegraphics[width=\textwidth]{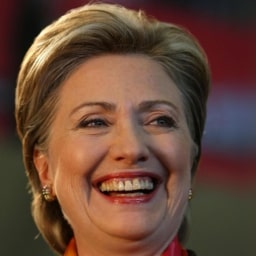}
    \end{subfigure}
    \begin{subfigure}[h]{0.16\textwidth}
        \caption*{Observed image}
        \centering
        \includegraphics[width=\textwidth]{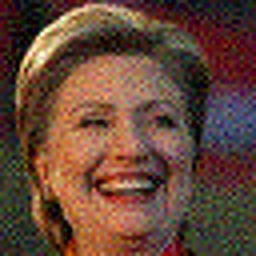}
    \end{subfigure}
    \begin{subfigure}[h]{0.16\textwidth}
        \caption*{SwinIR}
        \centering
        \includegraphics[width=\textwidth]{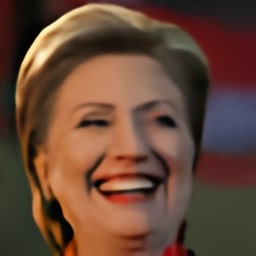}
    \end{subfigure}
    \begin{subfigure}[h]{0.16\textwidth}
        \caption*{DDRM}
        \centering
        \includegraphics[width=\textwidth]{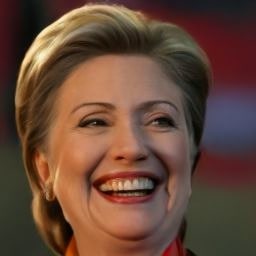}
    \end{subfigure}
    \\
    \begin{subfigure}[h]{0.16\textwidth}
        \caption*{DPS}
        \centering
        \includegraphics[width=\textwidth]{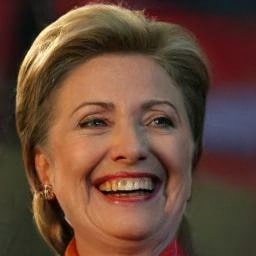}
    \end{subfigure}    
    \begin{subfigure}[h]{0.16\textwidth}
        \caption*{DiffPIR}
        \centering
        \includegraphics[width=\textwidth]{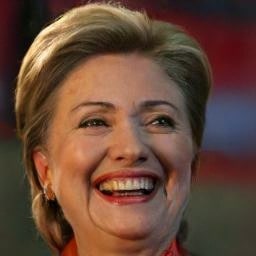}
    \end{subfigure}
    \begin{subfigure}[h]{0.16\textwidth}
        \caption*{{\bf IDPG}}
        \centering
        \includegraphics[width=\textwidth]{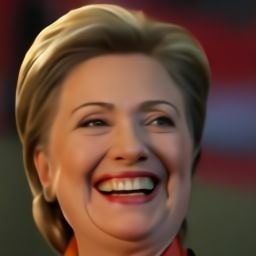}
    \end{subfigure}
    \begin{subfigure}[h]{0.16\textwidth}
        \caption*{{\bf DDPG}}
        \centering
        \includegraphics[width=\textwidth]{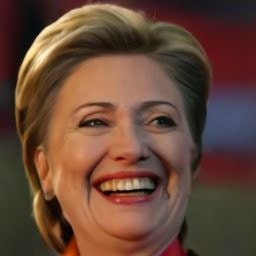}
    \end{subfigure}
    \caption{CelebA-HQ: SRx4 for bicubic downsampling with noise level 0.05.}
    \label{fig:SRx4_0.05_celeba_sup}
\end{figure}

\begin{figure}
    \centering
    \begin{subfigure}[h]{0.16\textwidth}
        \caption*{Ground truth}
        \centering
        \includegraphics[width=\textwidth]{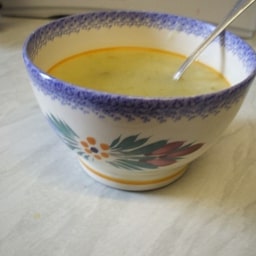}
    \end{subfigure}
    \begin{subfigure}[h]{0.16\textwidth}
        \caption*{Observed image}
        \centering
        \includegraphics[width=\textwidth]{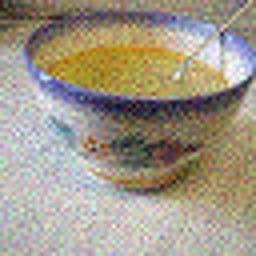}
    \end{subfigure}
    \begin{subfigure}[h]{0.16\textwidth}
        \caption*{SwinIR}
        \centering
        \includegraphics[width=\textwidth]{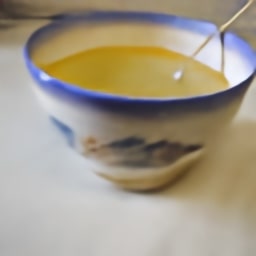}
    \end{subfigure}
    \begin{subfigure}[h]{0.16\textwidth}
        \caption*{DDRM}
        \centering
        \includegraphics[width=\textwidth]{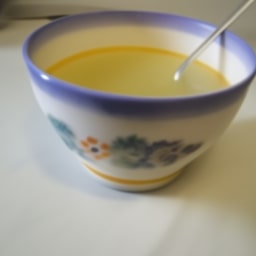}
    \end{subfigure}
    \\
    \begin{subfigure}[h]{0.16\textwidth}
        \caption*{DPS}
        \centering
        \includegraphics[width=\textwidth]{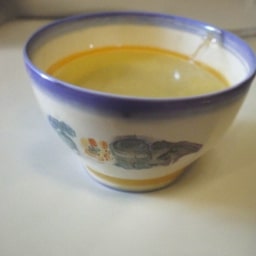}
    \end{subfigure}    
    \begin{subfigure}[h]{0.16\textwidth}
        \caption*{DiffPIR}
        \centering
        \includegraphics[width=\textwidth]{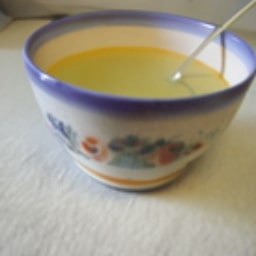}
    \end{subfigure}
    \begin{subfigure}[h]{0.16\textwidth}
        \caption*{{\bf IDPG}}
        \centering
        \includegraphics[width=\textwidth]{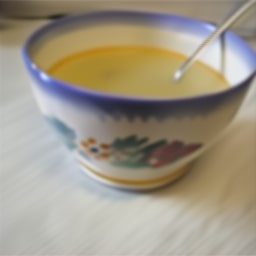}
    \end{subfigure}
    \begin{subfigure}[h]{0.16\textwidth}
        \caption*{{\bf DDPG}}
        \centering
        \includegraphics[width=\textwidth]{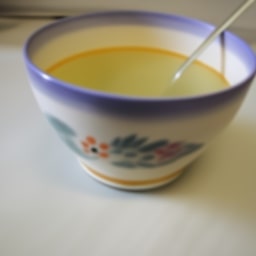}
    \end{subfigure}
    \\
    \vspace{0.2cm} % 
    \caption{ImageNet: SRx4 for bicubic downsampling with noise level 0.05.}
    \label{fig:SRx4_0.05_imagenet}
\end{figure}

\begin{figure}
    \centering
    \begin{subfigure}[h]{0.16\textwidth}
        \caption*{Ground truth}
        \centering
        \includegraphics[width=\textwidth]{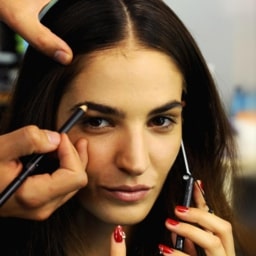}
    \end{subfigure}
    \begin{subfigure}[h]{0.16\textwidth}
        \caption*{Observed image}
        \centering
        \includegraphics[width=\textwidth]{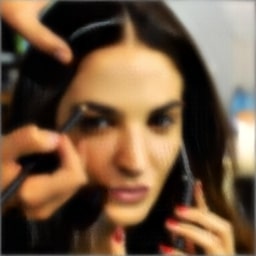}
    \end{subfigure}
    \begin{subfigure}[h]{0.16\textwidth}
        \caption*{Restormer}
        \centering
        \includegraphics[width=\textwidth]{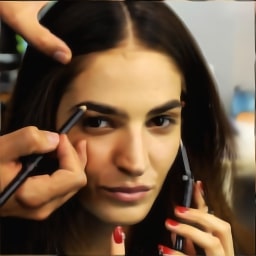}
    \end{subfigure}
    \begin{subfigure}[h]{0.16\textwidth}
        \caption*{DDRM}
        \centering
        \includegraphics[width=\textwidth]{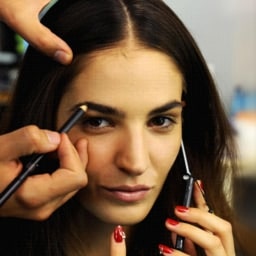}
    \end{subfigure} 
    \begin{subfigure}[h]{0.16\textwidth}
        \caption*{DDNM}
        \centering
        \includegraphics[width=\textwidth]{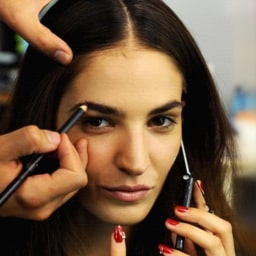}
    \end{subfigure} 
    \\
    \begin{subfigure}[h]{0.16\textwidth}
        \caption*{DPS}
        \centering
        \includegraphics[width=\textwidth]{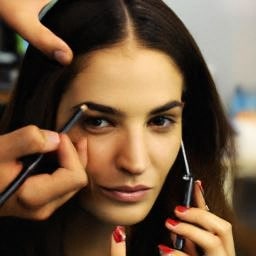}
    \end{subfigure}
    \begin{subfigure}[h]{0.16\textwidth}
        \caption*{DiffPIR}
        \centering
        \includegraphics[width=\textwidth]{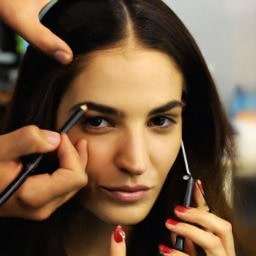}
    \end{subfigure}
    \begin{subfigure}[h]{0.16\textwidth}
        \caption*{{\bf IDPG}}
        \centering
        \includegraphics[width=\textwidth]{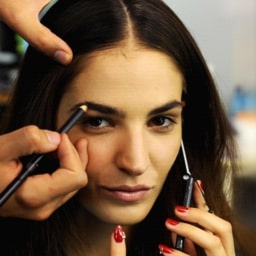}
    \end{subfigure}
    \begin{subfigure}[h]{0.16\textwidth}
        \caption*{{\bf DDPG}}
        \centering
        \includegraphics[width=\textwidth]{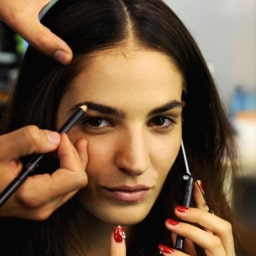}
    \end{subfigure}
    \caption{CelebA-HQ: Deblurring for noiseless Gaussian blur.}
    \label{fig:G_deb_0_celeba}
\end{figure}

\begin{figure}
    \centering
    \begin{subfigure}[h]{0.16\textwidth}
        \caption*{Ground truth}
        \centering
        \includegraphics[width=\textwidth]{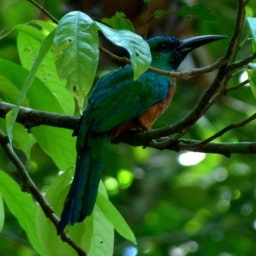}
    \end{subfigure}
    \begin{subfigure}[h]{0.16\textwidth}
        \caption*{Observed image}
        \centering
        \includegraphics[width=\textwidth]{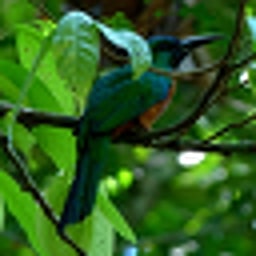}
    \end{subfigure}
    \begin{subfigure}[h]{0.16\textwidth}
        \caption*{DDRM}
        \centering
        \includegraphics[width=\textwidth]{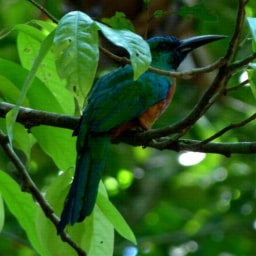}
    \end{subfigure} 
    \begin{subfigure}[h]{0.16\textwidth}
        \caption*{DDNM}
        \centering
        \includegraphics[width=\textwidth]{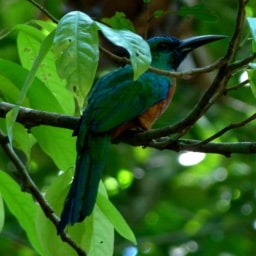}
    \end{subfigure} 
    \\
    \begin{subfigure}[h]{0.16\textwidth}
        \caption*{DPS}
        \centering
        \includegraphics[width=\textwidth]{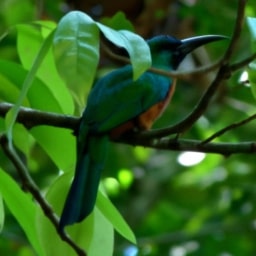}
    \end{subfigure}
    \begin{subfigure}[h]{0.16\textwidth}
        \caption*{DiffPIR}
        \centering
        \includegraphics[width=\textwidth]{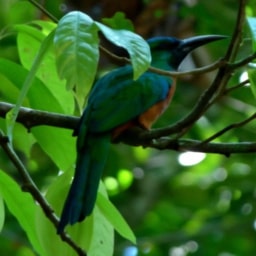}
    \end{subfigure}
    \begin{subfigure}[h]{0.16\textwidth}
        \caption*{{\bf IDPG}}
        \centering
        \includegraphics[width=\textwidth]{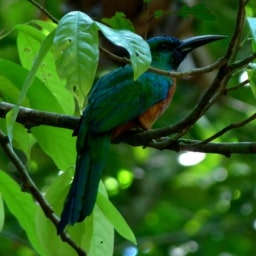}
    \end{subfigure}
    \begin{subfigure}[h]{0.16\textwidth}
        \caption*{{\bf DDPG}}
        \centering
        \includegraphics[width=\textwidth]{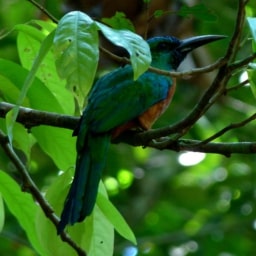}
    \end{subfigure}
    \caption{ImageNet: Deblurring for noiseless Gaussian blur.}
    \label{fig:G_deb_0_imagenet}
\end{figure}

\begin{figure}
    \centering
    \begin{subfigure}[h]{0.16\textwidth}
        \caption*{Ground truth}
        \centering
        \includegraphics[width=\textwidth]{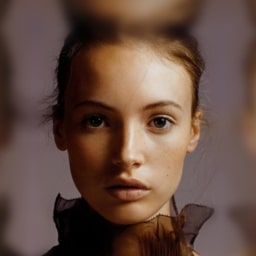}
    \end{subfigure}
    \begin{subfigure}[h]{0.16\textwidth}
        \caption*{Observed image}
        \centering
        \includegraphics[width=\textwidth]{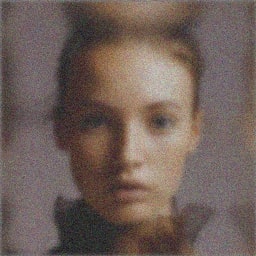}
    \end{subfigure}
    \begin{subfigure}[h]{0.16\textwidth}
        \caption*{Restormer}
        \centering
        \includegraphics[width=\textwidth]{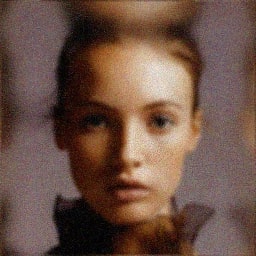}
    \end{subfigure}
    \begin{subfigure}[h]{0.16\textwidth}
        \caption*{DDRM}
        \centering
        \includegraphics[width=\textwidth]{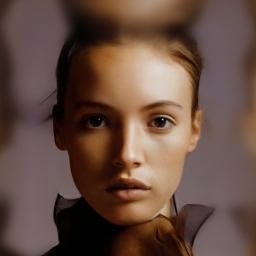}
    \end{subfigure}
    \\
    \begin{subfigure}[h]{0.16\textwidth}
        \caption*{DPS}
        \centering
        \includegraphics[width=\textwidth]{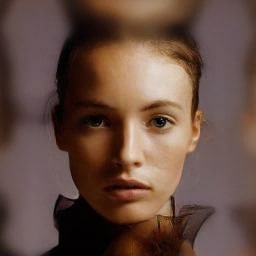}
    \end{subfigure}    
    \begin{subfigure}[h]{0.16\textwidth}
        \caption*{DiffPIR}
        \centering
        \includegraphics[width=\textwidth]{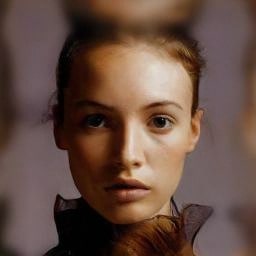}
    \end{subfigure}
    \begin{subfigure}[h]{0.16\textwidth}
        \caption*{{\bf IDPG}}
        \centering
        \includegraphics[width=\textwidth]{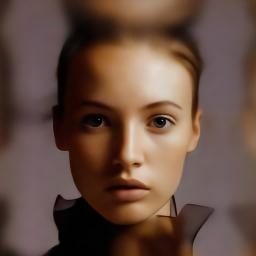}
    \end{subfigure}
    \begin{subfigure}[h]{0.16\textwidth}
        \caption*{{\bf DDPG}}
        \centering
        \includegraphics[width=\textwidth]{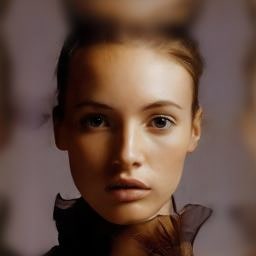}
    \end{subfigure}
    \caption{CelebA-HQ: Deblurring for Gaussian blur with noise level 0.05.}
    \label{fig:G_deb_0.05_celeba}
\end{figure}

\begin{figure}
    \centering
    \begin{subfigure}[h]{0.16\textwidth}
        \caption*{Ground truth}
        \centering
        \includegraphics[width=\textwidth]{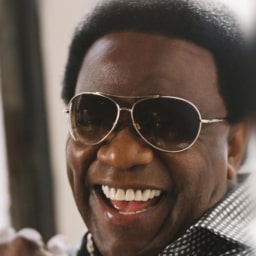}
    \end{subfigure}
    \begin{subfigure}[h]{0.16\textwidth}
        \caption*{Observed image}
        \centering
        \includegraphics[width=\textwidth]{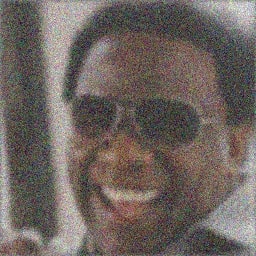}
    \end{subfigure}
    \begin{subfigure}[h]{0.16\textwidth}
        \caption*{Restormer}
        \centering
        \includegraphics[width=\textwidth]{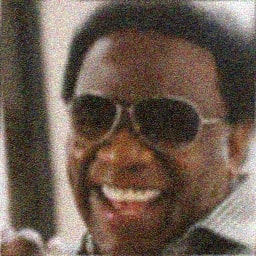}
    \end{subfigure}
    \begin{subfigure}[h]{0.16\textwidth}
        \caption*{DDRM}
        \centering
        \includegraphics[width=\textwidth]{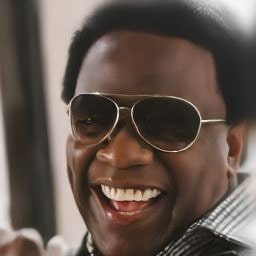}
    \end{subfigure}
    \\
    \begin{subfigure}[h]{0.16\textwidth}
        \caption*{DPS}
        \centering
        \includegraphics[width=\textwidth]{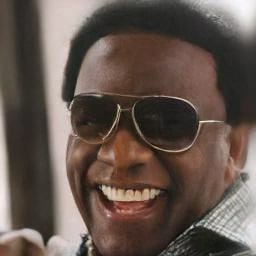}
    \end{subfigure}    
    \begin{subfigure}[h]{0.16\textwidth}
        \caption*{DiffPIR}
        \centering
        \includegraphics[width=\textwidth]{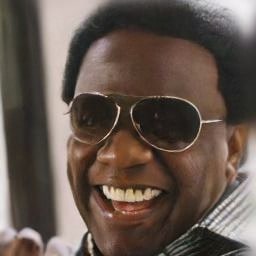}
    \end{subfigure}
    \begin{subfigure}[h]{0.16\textwidth}
        \caption*{{\bf IDPG}}
        \centering
        \includegraphics[width=\textwidth]{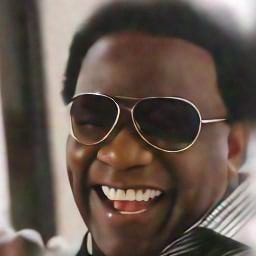}
    \end{subfigure}
    \begin{subfigure}[h]{0.16\textwidth}
        \caption*{{\bf DDPG}}
        \centering
        \includegraphics[width=\textwidth]{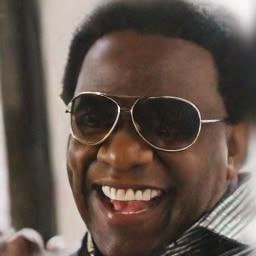}
    \end{subfigure}
    \caption{CelebA-HQ: Deblurring for Gaussian blur with noise level 0.1.}
    \label{fig:G_deb_0.1_celeba}
\end{figure}

\begin{figure}
    \centering
    \begin{subfigure}[h]{0.16\textwidth}
        \caption*{Ground truth}
        \centering
        \includegraphics[width=\textwidth]{figs/imagenet_deblur_g_0.05/188/188_gt.jpg}
    \end{subfigure}
    \begin{subfigure}[h]{0.16\textwidth}
        \caption*{Observed image}
        \centering
        \includegraphics[width=\textwidth]{figs/imagenet_deblur_g_0.05/188/188_y.jpg}
    \end{subfigure}
    \begin{subfigure}[h]{0.16\textwidth}
        \caption*{DDRM}
        \centering
        \includegraphics[width=\textwidth]{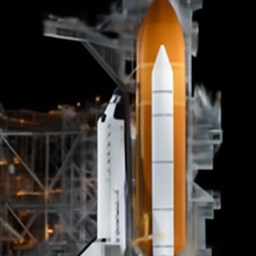}
    \end{subfigure}    
    \\
    \begin{subfigure}[h]{0.16\textwidth}
        \caption*{DPS}
        \centering
        \includegraphics[width=\textwidth]{figs/imagenet_deblur_g_0.05/188/188_DPS.jpg}
    \end{subfigure}
    \begin{subfigure}[h]{0.16\textwidth}
        \caption*{DiffPIR}
        \centering
        \includegraphics[width=\textwidth]{figs/imagenet_deblur_g_0.05/188/188_DIFFPIR.jpg}
    \end{subfigure}
    \begin{subfigure}[h]{0.16\textwidth}
        \caption*{{\bf IDPG}}
        \centering
        \includegraphics[width=\textwidth]{figs/imagenet_deblur_g_0.05/188/188_IDPG.jpg}
    \end{subfigure}
    \begin{subfigure}[h]{0.16\textwidth}
        \caption*{{\bf DDPG}}
        \centering
        \includegraphics[width=\textwidth]{figs/imagenet_deblur_g_0.05/188/188_DDPG_SSM_final.jpg}
    \end{subfigure}
    \\
    \vspace{0.2cm} % 
    \centering
    \begin{subfigure}[h]{0.16\textwidth}
        \caption*{Ground truth}
        \centering
        \includegraphics[width=\textwidth]{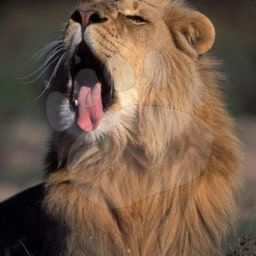}
    \end{subfigure}
    \begin{subfigure}[h]{0.16\textwidth}
        \caption*{Observed image}
        \centering
        \includegraphics[width=\textwidth]{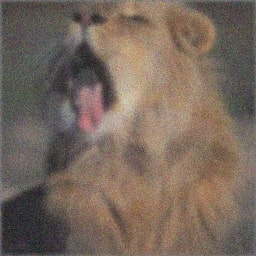}
    \end{subfigure}
    \begin{subfigure}[h]{0.16\textwidth}
        \caption*{DDRM}
        \centering
        \includegraphics[width=\textwidth]{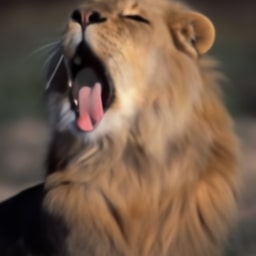}
    \end{subfigure}    
    \\
    \begin{subfigure}[h]{0.16\textwidth}
        \caption*{DPS}
        \centering
        \includegraphics[width=\textwidth]{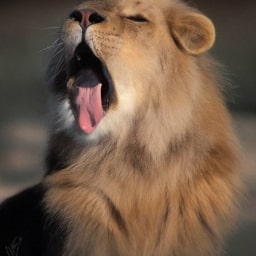}
    \end{subfigure}
    \begin{subfigure}[h]{0.16\textwidth}
        \caption*{DiffPIR}
        \centering
        \includegraphics[width=\textwidth]{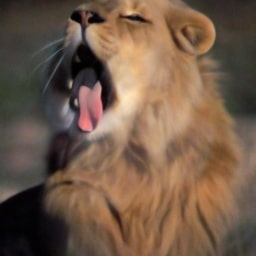}
    \end{subfigure}
    \begin{subfigure}[h]{0.16\textwidth}
        \caption*{{\bf IDPG}}
        \centering
        \includegraphics[width=\textwidth]{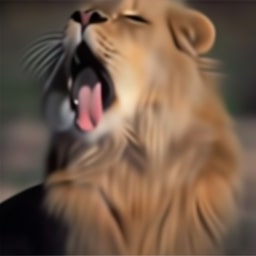}
    \end{subfigure}
    \begin{subfigure}[h]{0.16\textwidth}
        \caption*{{\bf DDPG}}
        \centering
        \includegraphics[width=\textwidth]{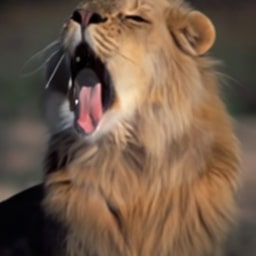}
    \end{subfigure}
    \caption{ImageNet: Deblurring for Gaussian blur with noise level 0.05.}
    \label{fig:G_deb_0.05_imagenet}
\end{figure}

\begin{figure}
    \centering
    \begin{subfigure}[h]{0.16\textwidth}
        \caption*{Ground truth}
        \centering
        \includegraphics[width=\textwidth]{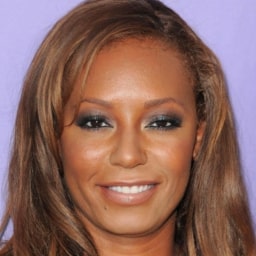}
    \end{subfigure}
    \begin{subfigure}[h]{0.16\textwidth}
        \caption*{Observed image}
        \centering
        \includegraphics[width=\textwidth]{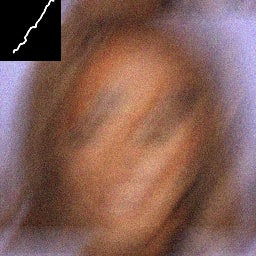}
    \end{subfigure}
    \begin{subfigure}[h]{0.16\textwidth}
        \caption*{DPS}
        \centering
        \includegraphics[width=\textwidth]{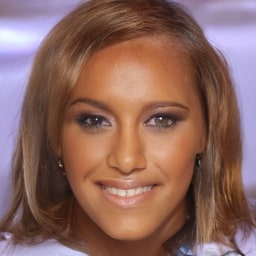}
    \end{subfigure}
    \begin{subfigure}[h]{0.16\textwidth}
        \caption*{DiffPIR}
        \centering
        \includegraphics[width=\textwidth]{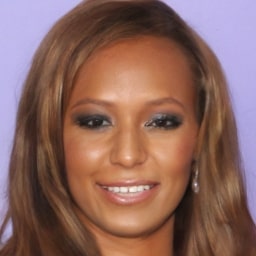}
    \end{subfigure}
    \begin{subfigure}[h]{0.16\textwidth}
        \caption*{{\bf IDPG}}
        \centering
        \includegraphics[width=\textwidth]{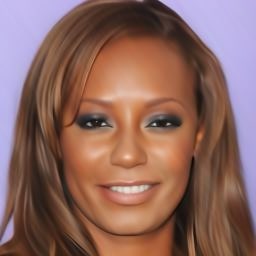}
    \end{subfigure}
    \begin{subfigure}[h]{0.16\textwidth}
        \caption*{{\bf DDPG}}
        \centering
        \includegraphics[width=\textwidth]{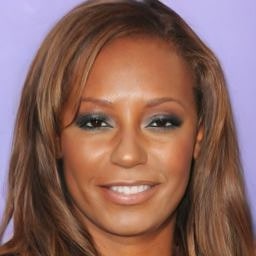}
    \end{subfigure}
    \\
    \vspace{0.2cm} % 
    \centering
    \begin{subfigure}[h]{0.16\textwidth}
        \caption*{Ground truth}
        \centering
        \includegraphics[width=\textwidth]{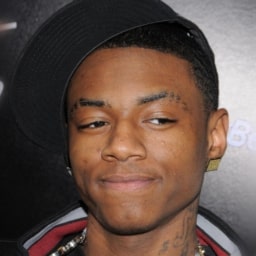}
    \end{subfigure}
    \begin{subfigure}[h]{0.16\textwidth}
        \caption*{Observed image}
        \centering
        \includegraphics[width=\textwidth]{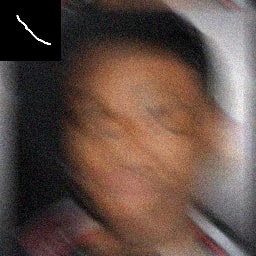}
    \end{subfigure}
    \begin{subfigure}[h]{0.16\textwidth}
        \caption*{DPS}
        \centering
        \includegraphics[width=\textwidth]{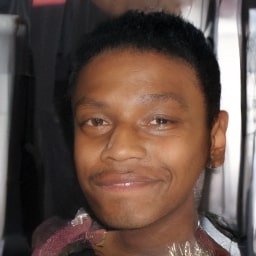}
    \end{subfigure}
    \begin{subfigure}[h]{0.16\textwidth}
        \caption*{DiffPIR}
        \centering
        \includegraphics[width=\textwidth]{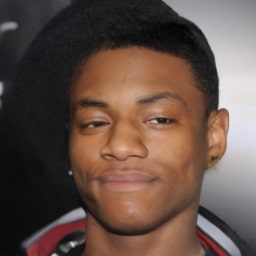}
    \end{subfigure}
    \begin{subfigure}[h]{0.16\textwidth}
        \caption*{{\bf IDPG}}
        \centering
        \includegraphics[width=\textwidth]{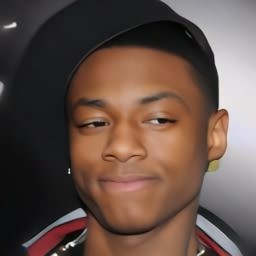}
    \end{subfigure}
    \begin{subfigure}[h]{0.16\textwidth}
        \caption*{{\bf DDPG}}
        \centering
        \includegraphics[width=\textwidth]{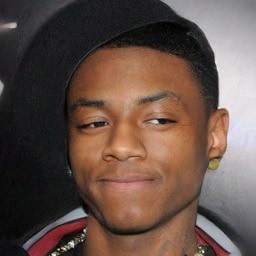}
    \end{subfigure}
    \caption{CelebA-HQ: Deblurring for motion blur with noise level 0.05.}
    \label{fig:Motion_deb_0.05_celeba}
\end{figure}

\begin{figure}
    \centering
    \begin{subfigure}[h]{0.16\textwidth}
        \caption*{Ground truth}
        \centering
        \includegraphics[width=\textwidth]{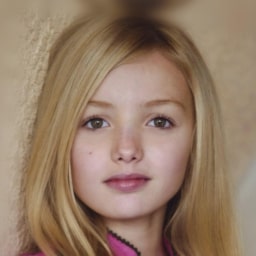}
    \end{subfigure}
    \begin{subfigure}[h]{0.16\textwidth}
        \caption*{Observed image}
        \centering
        \includegraphics[width=\textwidth]{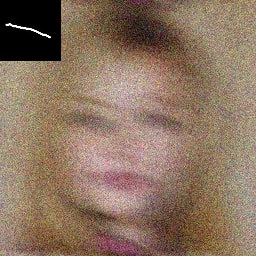}
    \end{subfigure}
    \begin{subfigure}[h]{0.16\textwidth}
        \caption*{DPS}
        \centering
        \includegraphics[width=\textwidth]{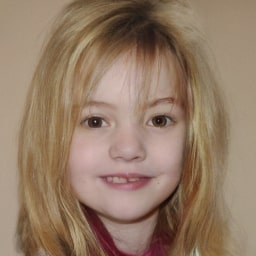}
    \end{subfigure}
    \begin{subfigure}[h]{0.16\textwidth}
        \caption*{DiffPIR}
        \centering
        \includegraphics[width=\textwidth]{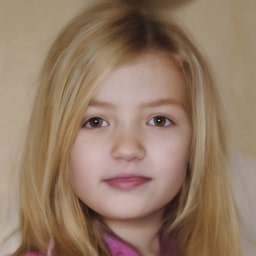}
    \end{subfigure}
    \begin{subfigure}[h]{0.16\textwidth}
        \caption*{{\bf IDPG}}
        \centering
        \includegraphics[width=\textwidth]{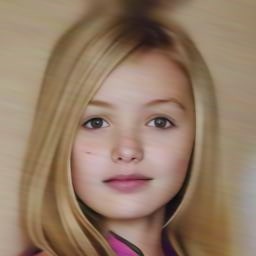}
    \end{subfigure}
    \begin{subfigure}[h]{0.16\textwidth}
        \caption*{{\bf DDPG}}
        \centering
        \includegraphics[width=\textwidth]{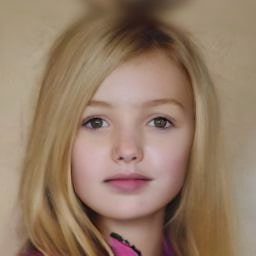}
    \end{subfigure}
    \\
    \vspace{0.2cm} % 
    \centering
    \begin{subfigure}[h]{0.16\textwidth}
        \caption*{Ground truth}
        \centering
        \includegraphics[width=\textwidth]{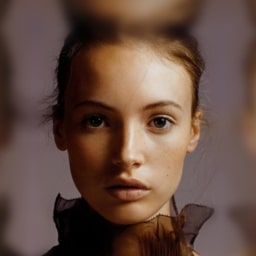}
    \end{subfigure}
    \begin{subfigure}[h]{0.16\textwidth}
        \caption*{Observed image}
        \centering
        \includegraphics[width=\textwidth]{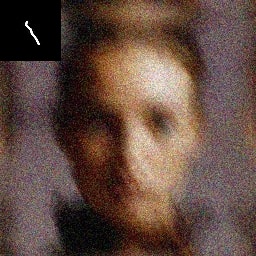}
    \end{subfigure}
    \begin{subfigure}[h]{0.16\textwidth}
        \caption*{DPS}
        \centering
        \includegraphics[width=\textwidth]{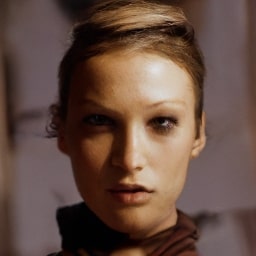}
    \end{subfigure}
    \begin{subfigure}[h]{0.16\textwidth}
        \caption*{DiffPIR}
        \centering
        \includegraphics[width=\textwidth]{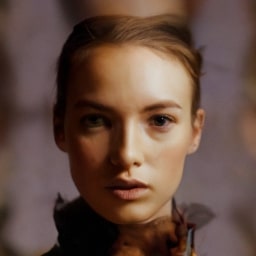}
    \end{subfigure}
    \begin{subfigure}[h]{0.16\textwidth}
        \caption*{{\bf IDPG}}
        \centering
        \includegraphics[width=\textwidth]{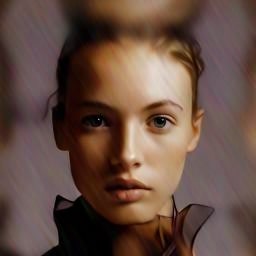}
    \end{subfigure}
    \begin{subfigure}[h]{0.16\textwidth}
        \caption*{{\bf DDPG}}
        \centering
        \includegraphics[width=\textwidth]{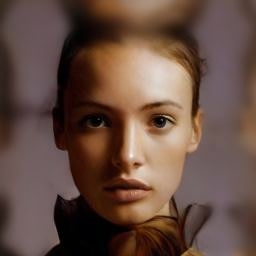}
    \end{subfigure}
        \\
    \vspace{0.2cm} % 
    \centering
    \begin{subfigure}[h]{0.16\textwidth}
        \caption*{Ground truth}
        \centering
        \includegraphics[width=\textwidth]{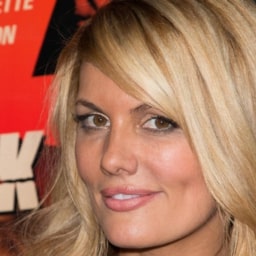}
    \end{subfigure}
    \begin{subfigure}[h]{0.16\textwidth}
        \caption*{Observed image}
        \centering
        \includegraphics[width=\textwidth]{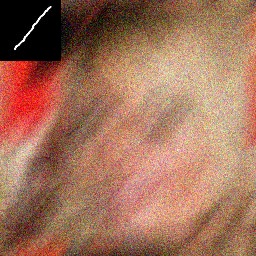}
    \end{subfigure}
    \begin{subfigure}[h]{0.16\textwidth}
        \caption*{DPS}
        \centering
        \includegraphics[width=\textwidth]{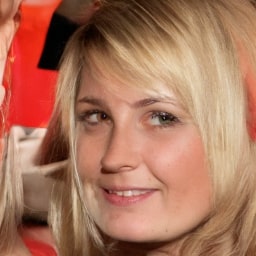}
    \end{subfigure}
    \begin{subfigure}[h]{0.16\textwidth}
        \caption*{DiffPIR}
        \centering
        \includegraphics[width=\textwidth]{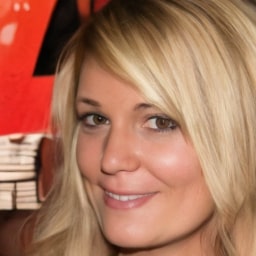}
    \end{subfigure}
    \begin{subfigure}[h]{0.16\textwidth}
        \caption*{{\bf IDPG}}
        \centering
        \includegraphics[width=\textwidth]{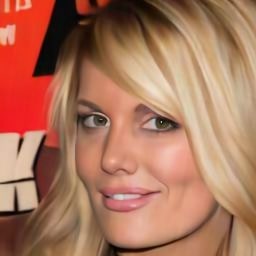}
    \end{subfigure}
    \begin{subfigure}[h]{0.16\textwidth}
        \caption*{{\bf DDPG}}
        \centering
        \includegraphics[width=\textwidth]{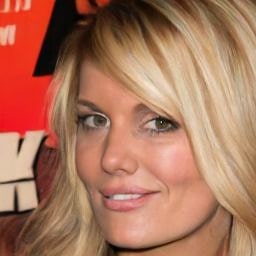}
    \end{subfigure}
    \caption{CelebA-HQ: Deblurring for motion blur with noise level 0.1.}
    \label{fig:Motion_deb_0.1_celeba}
\end{figure}

\begin{figure}
    \centering
    \begin{subfigure}[h]{0.16\textwidth}
        \caption*{Ground truth}
        \centering
        \includegraphics[width=\textwidth]{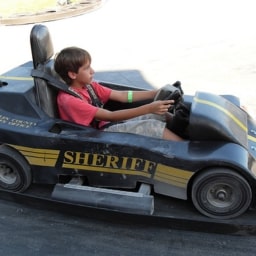}
    \end{subfigure}
    \begin{subfigure}[h]{0.16\textwidth}
        \caption*{Observed image}
        \centering
        \includegraphics[width=\textwidth]{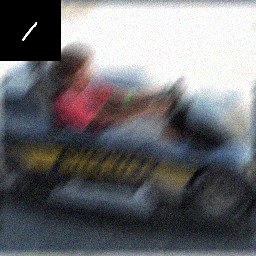}
    \end{subfigure}
    \begin{subfigure}[h]{0.16\textwidth}
        \caption*{DPS}
        \centering
        \includegraphics[width=\textwidth]{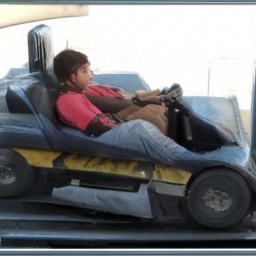}
    \end{subfigure}
    \begin{subfigure}[h]{0.16\textwidth}
        \caption*{DiffPIR}
        \centering
        \includegraphics[width=\textwidth]{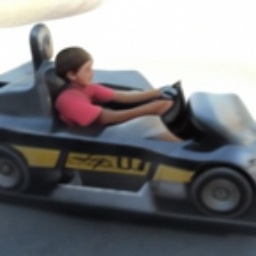}
    \end{subfigure}
    \begin{subfigure}[h]{0.16\textwidth}
        \caption*{{\bf IDPG}}
        \centering
        \includegraphics[width=\textwidth]{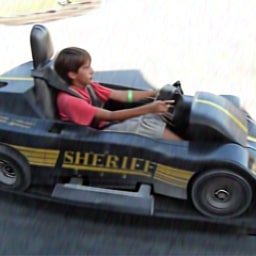}
    \end{subfigure}
    \begin{subfigure}[h]{0.16\textwidth}
        \caption*{{\bf DDPG}}
        \centering
        \includegraphics[width=\textwidth]{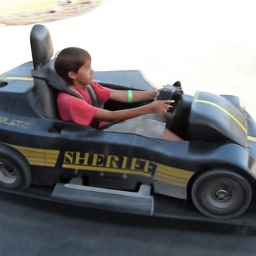}
    \end{subfigure}
    \\
    \vspace{0.2cm} % 
    \centering
    \begin{subfigure}[h]{0.16\textwidth}
        \caption*{Ground truth}
        \centering
        \includegraphics[width=\textwidth]{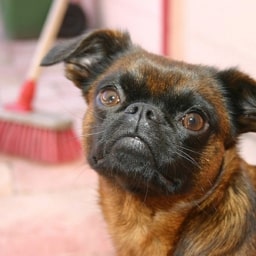}
    \end{subfigure}
    \begin{subfigure}[h]{0.16\textwidth}
        \caption*{Observed image}
        \centering
        \includegraphics[width=\textwidth]{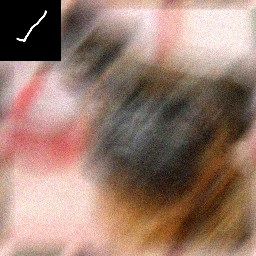}
    \end{subfigure}
    \begin{subfigure}[h]{0.16\textwidth}
        \caption*{DPS}
        \centering
        \includegraphics[width=\textwidth]{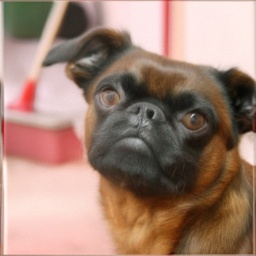}
    \end{subfigure}
    \begin{subfigure}[h]{0.16\textwidth}
        \caption*{DiffPIR}
        \centering
        \includegraphics[width=\textwidth]{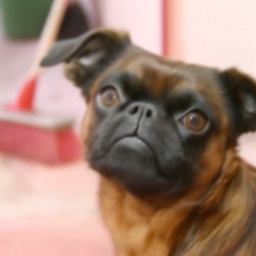}
    \end{subfigure}
    \begin{subfigure}[h]{0.16\textwidth}
        \caption*{{\bf IDPG}}
        \centering
        \includegraphics[width=\textwidth]{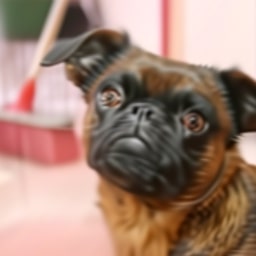}
    \end{subfigure}
    \begin{subfigure}[h]{0.16\textwidth}
        \caption*{{\bf DDPG}}
        \centering
        \includegraphics[width=\textwidth]{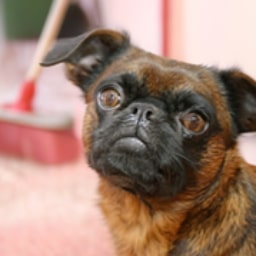}
    \end{subfigure}
        \vspace{0.2cm} % 
    \centering
    \begin{subfigure}[h]{0.16\textwidth}
        \caption*{Ground truth}
        \centering
        \includegraphics[width=\textwidth]{figs/imagenet_motion_deb_0.05/4/4_gt.jpg}
    \end{subfigure}
    \begin{subfigure}[h]{0.16\textwidth}
        \caption*{Observed image}
        \centering
        \includegraphics[width=\textwidth]{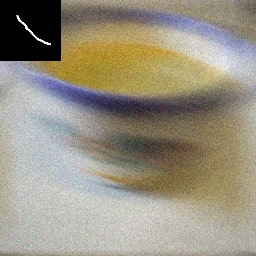}
    \end{subfigure}
    \begin{subfigure}[h]{0.16\textwidth}
        \caption*{DPS}
        \centering
        \includegraphics[width=\textwidth]{figs/imagenet_motion_deb_0.05/4/4_DPS.JPEG}
    \end{subfigure}
    \begin{subfigure}[h]{0.16\textwidth}
        \caption*{DiffPIR}
        \centering
        \includegraphics[width=\textwidth]{figs/imagenet_motion_deb_0.05/4/4_DIFFPIR.JPEG}
    \end{subfigure}
    \begin{subfigure}[h]{0.16\textwidth}
        \caption*{{\bf IDPG}}
        \centering
        \includegraphics[width=\textwidth]{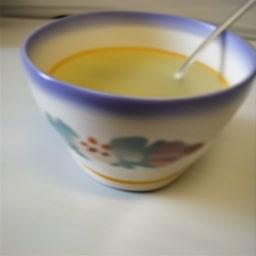}
    \end{subfigure}
    \begin{subfigure}[h]{0.16\textwidth}
        \caption*{{\bf DDPG}}
        \centering
        \includegraphics[width=\textwidth]{figs/imagenet_motion_deb_0.05/4/4_DDPG_SSM.jpg}
    \end{subfigure}
    \caption{ImageNet: Deblurring for motion blur with noise level 0.05.}
    \label{fig:Motion_deb_0.05_imagenet}
\end{figure}

\begin{figure}
    \centering
    \begin{subfigure}[h]{0.16\textwidth}
        \caption*{Ground truth}
        \centering
        \includegraphics[width=\textwidth]{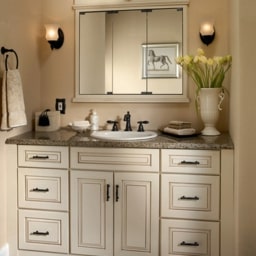}
    \end{subfigure}
    \begin{subfigure}[h]{0.16\textwidth}
        \caption*{Observed image}
        \centering
        \includegraphics[width=\textwidth]{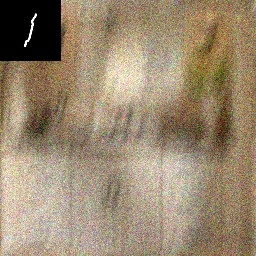}
    \end{subfigure}
    \begin{subfigure}[h]{0.16\textwidth}
        \caption*{DPS}
        \centering
        \includegraphics[width=\textwidth]{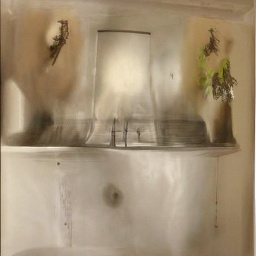}
    \end{subfigure}
    \begin{subfigure}[h]{0.16\textwidth}
        \caption*{DiffPIR}
        \centering
        \includegraphics[width=\textwidth]{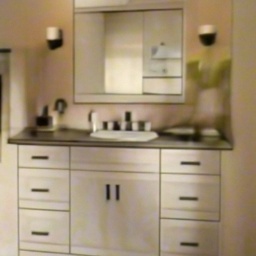}
    \end{subfigure}
    \begin{subfigure}[h]{0.16\textwidth}
        \caption*{{\bf IDPG}}
        \centering
        \includegraphics[width=\textwidth]{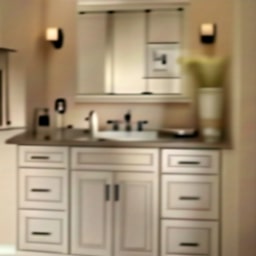}
    \end{subfigure}
    \begin{subfigure}[h]{0.16\textwidth}
        \caption*{{\bf DDPG}}
        \centering
        \includegraphics[width=\textwidth]{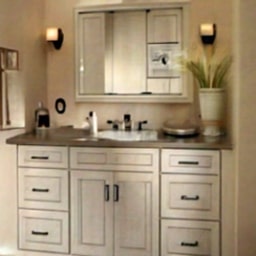}
    \end{subfigure}
    \\
    \vspace{0.2cm} % 
    \centering
    \begin{subfigure}[h]{0.16\textwidth}
        \caption*{Ground truth}
        \centering
        \includegraphics[width=\textwidth]{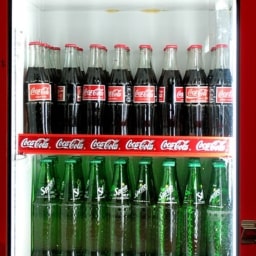}
    \end{subfigure}
    \begin{subfigure}[h]{0.16\textwidth}
        \caption*{Observed image}
        \centering
        \includegraphics[width=\textwidth]{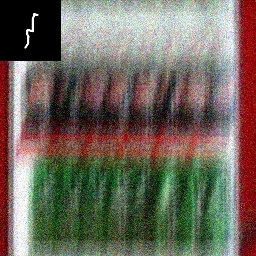}
    \end{subfigure}
    \begin{subfigure}[h]{0.16\textwidth}
        \caption*{DPS}
        \centering
        \includegraphics[width=\textwidth]{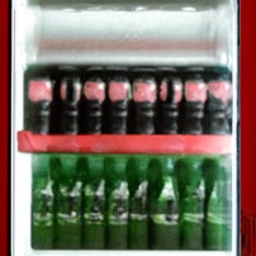}
    \end{subfigure}
    \begin{subfigure}[h]{0.16\textwidth}
        \caption*{DiffPIR}
        \centering
        \includegraphics[width=\textwidth]{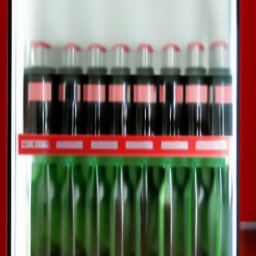}
    \end{subfigure}
    \begin{subfigure}[h]{0.16\textwidth}
        \caption*{{\bf IDPG}}
        \centering
        \includegraphics[width=\textwidth]{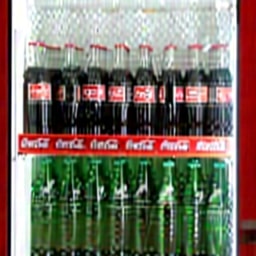}
    \end{subfigure}
    \begin{subfigure}[h]{0.16\textwidth}
        \caption*{{\bf DDPG}}
        \centering
        \includegraphics[width=\textwidth]{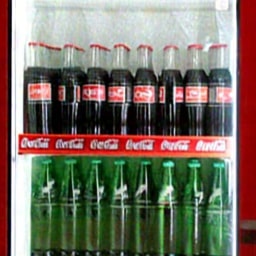}
    \end{subfigure}
    \caption{ImageNet: Deblurring for motion blur with noise level 0.1.}
    \label{fig:Motion_deb_0.1_imagenet}
\end{figure}

\end{document}